\documentclass[letterpaper, 11pt]{article}

\usepackage{graphicx}
\usepackage{lipsum} 
\usepackage{amsmath}
\usepackage{amsthm}
\usepackage{mathtools}
\usepackage{amsfonts}
\usepackage{amssymb}
\usepackage{empheq} 
\usepackage{dashbox}
\usepackage{booktabs}

\usepackage{enumitem}
\usepackage[letterpaper, margin=1in]{geometry}
\usepackage[section]{placeins}

\newtheorem{lemma}{Lemma}
\newtheorem{theorem}{Theorem}
\newtheorem{definition}{Definition}

\newtheorem{example}{Example}
\newtheorem{corollary}{Corollary}

\newcommand{\argmin}{\textup{argmin}}

\newcommand{\cU}{{\mathcal{U}}}

\newcommand{\cS}{{\mathcal{S}}}
\newcommand{\cF}{{\mathcal{F}}}

\newcommand{\diver}{\textup{div}}
\newcommand{\lsf}{\ell^*_{\text{fair}}}
\newcommand{\rsf}{r^*_{\text{fair}}}

\newcommand{\comb}{M}

\newcommand{\setA}{\cS^{+}}
\newcommand{\setB}{\cS^{-}}

\usepackage{xspace}
\newcommand{\gonzalez}{\textup{GMM}\xspace}
\newcommand{\fMM}{\textsc{Fair Max-Min}\xspace}
\newcommand{\fMMG}{\textsc{Fair$^{+}$ Max-Min}\xspace}
\newcommand{\fMMSwap}{\textsc{Fair-Swap}\xspace}
\newcommand{\fMMGSwap}{\textsc{Fair$^{+}$-Swap}\xspace}
\newcommand{\fMMFlow}{\textsc{Fair-Flow}\xspace}
\newcommand{\fMMGFlow}{\textsc{Fair$^{+}$-Flow}\xspace}
\newcommand{\fGMM}{\textsc{Fair-GMM}\xspace}
\newcommand{\fClust}{\textsc{Fair-Flow-Clust}\xspace}

\usepackage{optidef}
\usepackage{xcolor}
\newcommand\mycommfont[1]{\ttfamily\textcolor{gray}{#1}}

\usepackage{algorithm}
\usepackage[noend]{algpseudocode}
\usepackage{setspace}

\algnewcommand{\LineComment}[1]{\State{\mycommfont{$\triangleright${#1}}}}
\algnewcommand{\LineCommentx}[1]{\Statex{\mycommfont{$\triangleright${#1}}}}

\usepackage{colortbl}
\usepackage{adjustbox}
\usepackage{subcaption}
\DeclareMathOperator*{\maximize}{maximize}

\usepackage[linktocpage=true]{hyperref} 

\title{	
\LARGE Diverse Data Selection under Fairness Constraints\footnote{University of Massachusetts Amherst, \{zmoumoulidou, mcgregor, ameli\}@cs.umass.edu}
}

\author{
  \large Zafeiria Moumoulidou \qquad Andrew McGregor \qquad Alexandra Meliou\\
  } 
\date{} 

\begin{document}
	\maketitle
	
	\begin{abstract}
	Diversity is an important principle in data selection and summarization,
	facility location, and recommendation systems. Our work focuses on maximizing
	diversity in data selection, while offering fairness guarantees. In
	particular, we offer the first study that augments the Max-Min 	diversification objective with fairness constraints. More specifically, 	given a universe $\mathcal{U}$ of $n$ elements that can be partitioned into 	$m$ disjoint groups, we aim to retrieve a $k$-sized subset that maximizes the 	pairwise minimum distance within the set (\emph{diversity}) and contains a
	pre-specified $k_i$ number of elements from each group $i$ (\emph{fairness}).
	We show that this problem is NP-complete even in metric spaces, and we 	propose three novel algorithms, linear in $n$, that provide strong theoretical
	approximation guarantees for different values of $m$ and $k$. Finally, we
	extend our algorithms and analysis to the case where groups can be 	overlapping.
	\end{abstract}	
	
	\section{Introduction} \label{sec:intro} 
Data is generated and collected from all aspects of human activity, in domains
like commerce, medicine, and transportation, as well as scientific
measurements, simulations, and environmental monitoring. However, while
datasets grow large and are readily available, they are often down-sampled for
various uses. This is often due to practical implications, e.g., analytics
workflows may be designed, tested, and debugged over subsets of the data for
efficiency reasons. Other times, machine learning applications use subsets of
the data for training and testing, while applications that target human
consumption, e.g., data exploration, can only display small parts of the data
at a time, since human users can visually process limited information.

\begin{figure}[t]
	\centering
	\includegraphics[width=0.9\textwidth]{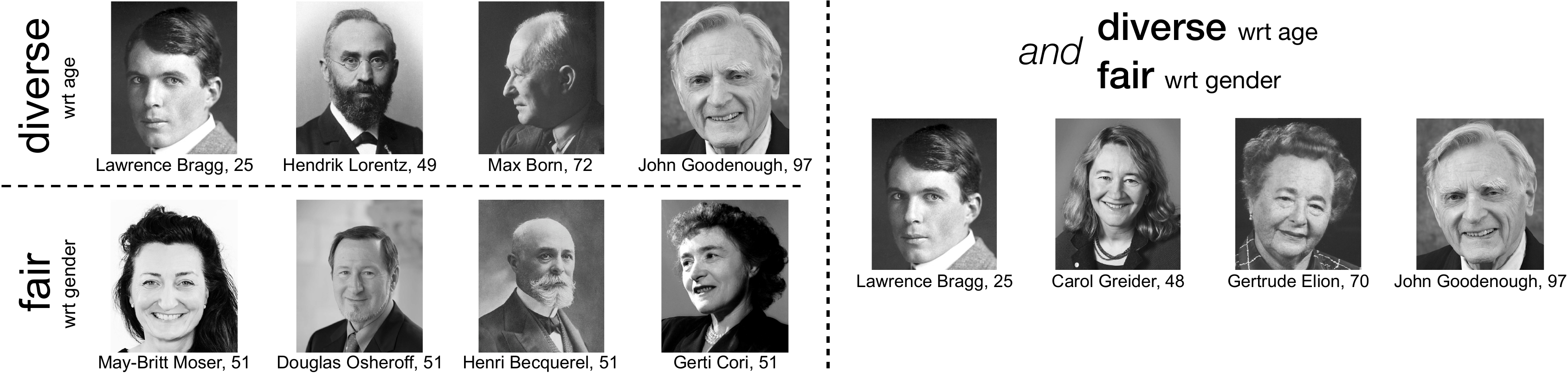}
	\caption{
	Three examples of selection of four items from a dataset of Nobel 	laureates. The first set on the left is diverse with respect to age; the 	second set 	is fair with respect to gender; the third set, on the right, is 	both diverse with respect 	to age and fair with respect to 	gender.}\label{fig:nobelists}
	\vspace{-2mm}
\end{figure} 

While data subset selection is very common, deriving \emph{good} subsets is a
non-trivial task. In this paper, we focus on two principles in data selection:
\emph{diversity} and \emph{fairness}. Diversity and fairness are related but
distinct concepts. Specifically, diversity seeks to maximize the dissimilarity
of the items in a set. Intuitively, a diverse set of items selected from a
dataset $D$ represents more and different aspects of the information present in
$D$. Prior work has suggested several diversity
objectives~\cite{Chandra:2001:AAD:374591.374596,
Hassin:1997:AAM:2308449.2308622, 10.1145/2594538.2594560,
1994:HSC:2753204.2753214}, typically defined in terms of an element-wise
distance function over numerical attributes (e.g., geographic location, age).
On the other hand, fairness aims to achieve some specified level of
representation across different categories or groups, and is typically defined
over categorical attributes (e.g., race, gender). While one could consider
combining fairness and diversity into a single objective, comparing numerical
and categorical attributes is not straightforward, as it typically requires ad
hoc decisions in discretizing numerical attributes or defining a distance
function involving numeric and categorical attributes.

Figure~\ref{fig:nobelists} demonstrates an example of the principles of
diversity and fairness in subset selection. Consider a web search query over a
dataset of Nobel laureates. There are close to a thousand laureates, but the web
search only serves a small number of results for human consumption.
Figure~\ref{fig:nobelists} shows three examples of possible subsets of four
items. The first subset optimizes the set's diversity with respect to the age
of the laureates at the time of the award, but only contains male scientists.
The second set achieves fair gender representation, but is not diverse with
respect to age. The third set achieves both diversity and fairness. The concept
of \emph{fair} and \emph{diverse} data selection is motivated by many
real-world scenarios: \emph{transportation equity} in conjunction with
optimizing traditional objectives (e.g., geographic coverage) aims to
design accessible transportation systems for historically disadvantaged
groups~\cite{Litman2006EvaluatingTE}; formulating \emph{teams} that represent
various demographic groups while demonstrating ``diversity of thought'' is
becoming an important hiring goal~\cite{diversityThought1,
article,diversityThought2}; in \emph{news} websites, a summary of dissimilar in
context documents from different news channels minimizes redundancy and
mitigates the risk of showing a polarized opinion~\cite{article}.

\smallskip

\noindent
\textbf{Our focus.} \looseness -1
In this paper, \emph{our goal is to maximize diversity in data selection} with
respect to numerical attributes, \emph{while ensuring the satisfaction of
fairness constraints} with respect to categorical ones. We focus on the Max-Min
diversification model~\cite{article, 1994:HSC:2753204.2753214, Tamir}, which is
among the most well-studied and frequently-used diversity models in the data
management community. Max-Min diversification seeks to select a set of $k$
items, such that the distance between any two items is maximized. We
further express fairness as cardinality constraints: given $m$ 
demographic groups, a set is \emph{fair} if it contains a
pre-specified integer number $k_i$ of representatives from each group. This
general form of cardinality constraints captures, among others, the common
fairness objectives of \emph{proportional} representation, where the sample
preserves the demographic proportions of the general population, and
\emph{equal} representation, where all demographic groups are equally
represented in the sample.\footnote{Our fairness constraints are based on the
definitions of \emph{group fairness} and \emph{statistical
parity}~\cite{Dwork:2012:FTA:2090236.2090255}. Other definitions that focus on
\emph{individual or causal fairness} examine differences in treatment of
individuals from different groups who are otherwise very similar, but these are
not the focus of this work.} 

We first study the problem of \emph{fair Max-Min diversification} in the case of \emph{non-overlapping} groups, and define the problem more formally as follows: We assume a universe of elements $\mathcal{U} = \bigcup_{i=1}^{m} \mathcal{U}_i$ partitioned into $m$ non-overlapping groups, a metric distance function $d$ defined for any two pairs of elements, and a set of fairness constraints $\langle k_1, k_2, \cdots, k_m\rangle $, where each $k_i$ is a non-negative integer with $k_i \leq |\mathcal{U}_i|$. Our goal is to select a set $\mathcal{S}\subseteq\mathcal{U}$ of size $k=\sum_{i=1}^{m}k_i$, such that $|\mathcal{S}\cap\mathcal{U}_i|=k_i$ for all $i$, and such that the minimum distance of any two items in $\mathcal{S}$ is maximized. In this paper, we show that fair Max-Min diversification is NP-complete, and we contribute efficient algorithms with strong approximation guarantees in the case of non-overlapping groups; we further generalize our results and analysis to the case of overlapping groups. We list our contributions at the end of this section.

\begin{figure}[t]
	\begin{subfigure}[b]{0.65\textwidth}
		\vspace{-10mm}
		\centering
		\resizebox{1.1\columnwidth}{!}{
			\def\arraystretch{1.5}
			\begin{tabular}{|r|c|c|}
				\hline
				& \textbf{Max-Min} & \textbf{Max-Sum}\\ 
				\hline
				\textbf{diversification} & \multicolumn{2}{c|}{
					\begin{tabular}[c]{@{}c@{}}
						\cite{Hassin:1997:AAM:2308449.2308622, 						1994:HSC:2753204.2753214, Tamir}\\ 
						$\tfrac{1}{2}$-approximation
					\end{tabular}}\\ 
				\hline
				\begin{tabular}[c]{@{}c@{}}
					\textbf{fair diversification} \\
					\textbf{(disjoint groups)}
				\end{tabular}
				&
				\cellcolor[gray]{0.9}
					\begin{tabular}[c]{@{}c@{}}
						[this paper]\\
						$\tfrac{1}{4}$-approx. ($m=2$)\\
						$\tfrac{1}{3m-1}$-approx. ($m\ge 3$)\\
						$\tfrac{1}{5}$-approx. ($m=O(1)$ and $k=o(\log n)$)
					\end{tabular} 
				&
				\begin{tabular}[c]{@{}c@{}}
					\cite{Abbassi:2013:DMU:2487575.2487636,Borodin:2012:MDM:2213556.2213580, Ceccarello:2018:FCD:3159652.3159719}\\ 					$\left(\tfrac{1}{2}-\epsilon\right)$-approx.
				\end{tabular}\\
				\hline
				\begin{tabular}[c]{@{}c@{}}
					\textbf{fair diversification} \\
					\textbf{(overlapping groups)}
				\end{tabular}
				&
				\cellcolor[gray]{0.9}
					\begin{tabular}[c]{@{}c@{}}
						[this paper]\\ 
		    			$\tfrac{1}{4}$-approx. ($m=2$)\\
						$\tfrac{1}{3\binom{m}{\lfloor m/2 \rfloor} -1}$-approx. 						($m\ge 3$)
					\end{tabular}
				&
				N/A \\
				\hline
				\multicolumn{3}{c}{\small{$n:$ \# elements in the universe,  				$m:$ \# demographic groups, $k:$ \# elements in the data 				selection task}}
			\end{tabular}
		}
		\vspace{-1mm}
		\captionsetup{justification=centering}
		\caption{Comparison with prior art}\label{fig:related}
		\vspace{-1mm}
	\end{subfigure}
	~
	\begin{subfigure}[b]{0.4\textwidth}
		\centering
		\includegraphics[width=0.65\textwidth]{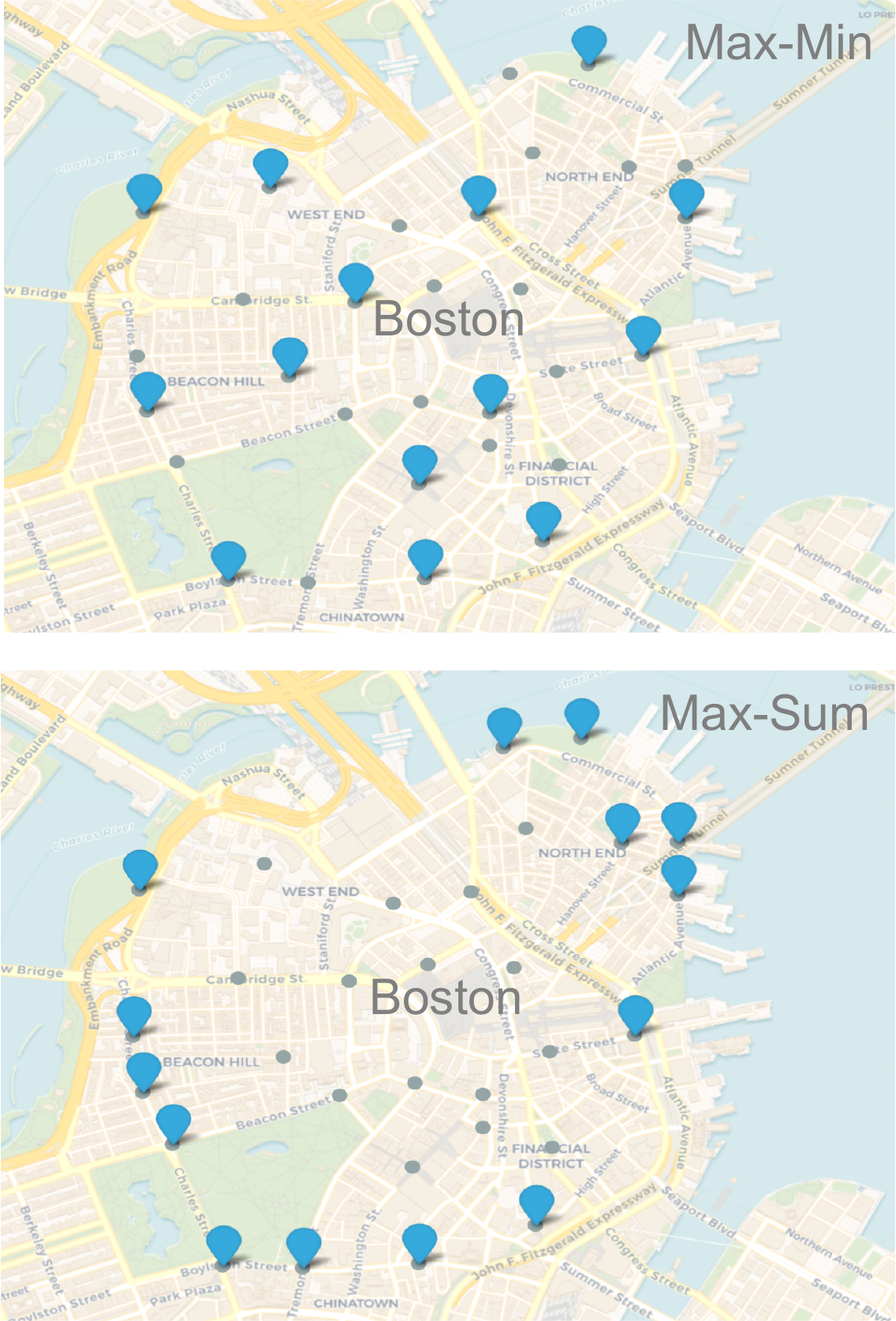}
		\vspace{-1mm}
		\captionsetup{justification=centering}
		\caption{ Max-Min vs Max-Sum}\label{fig:visual}
		\vspace{-1mm}
	\end{subfigure}
	\caption{(a)~Contributions of this paper with respect to the prior art. Our work is the first to introduce fairness constraints to Max-Min diversification, and provides strong approximation results. We also contribute algorithms to the case of overlapping classes, which has not been addressed in prior work. (b)~The department of transportation wants to place $k=14$ new bike 	sharing stations in downtown Boston among $n=30$ candidate locations. (Top):~Max-Min selects locations that geographically \emph{cover} downtown. (Bottom):~Max-Sum selects locations on the outskirts of downtown.}
	\vspace{-2mm}
\end{figure}

\subsection*{Contrast with prior work and related problems} 
Our work augments the existing literature of traditional problems that have been studied under \emph{group fairness} constraints, such as
clustering~\cite{Chierichetti:2017:FCT:3295222.3295256, pmlr-v97-kleindessner19a}, ranking systems~\cite{Celis2017RankingWF, Yang:2019:BRD:3367722.3367886, YangS2017, Zehlike:2017:FFT:3132847.3132938} and set selection~\cite{Stoyanovich2018OnlineSS}. We proceed to review prior work in
closely-related problems and describe how our contributions augment the
existing literature. (Summary shown in Figure~\ref{fig:related}.) 

\smallskip

\noindent 
\textbf{Max-Min and Max-Sum diversification.}
The unconstrained version of Max-Min diversification is a special case of our
fair variant for $m=1$. This problem was initially studied in the operation
research literature under the name \emph{remote-edge} or \emph{p-dispersion},
along with another popular diversity model, the \emph{Max-Sum} or
\emph{remote-clique} model~\cite{Chandra:2001:AAD:374591.374596, ERKUT199048,
Hassin:1997:AAM:2308449.2308622,doi:10.1111/j.1538-4632.1987.tb00133.x,1994:HSC:2753204.2753214}. Similar formulations have also been studied in the context of
obnoxious facility location on graphs~\cite{Tamir}. While the Max-Min model
aims to maximize the minimum pairwise distance in the selected set, the Max-Sum
model aims to maximize the total sum of pairwise distances in a set of $k$
items. Max-Sum, as an additive objective, is easier to analyze but tends to
select points at the limits of the data space and thus it is not well-suited to
applications that require more uniform coverage (see example in
Figure~\ref{fig:visual}). The unconstrained diversification problems are
NP-complete even in metric spaces but, for both, a greedy algorithm offers a
$\tfrac{1}{2}-$factor approximation, that has also been shown to be
tight~\cite{10.5555/3157382.3157556, 10.1145/3086464, 1994:HSC:2753204.2753214}.
 
\smallskip

\noindent 
\textbf{Fair Max-Sum diversification.} \looseness-1
Abbassi et al.~\cite{Abbassi:2013:DMU:2487575.2487636} study the \emph{fair}
Max-Sum diversification problem (assuming disjoint groups) under matroid
constraints, where the retrieved subset needs to be an independent set of a
matroid of size $k$ (we discuss the correspondence between group fairness
constraints and partition matroids in Section~\ref{sec:2.2}). They
propose a local search algorithm with a
$\left(\tfrac{1}{2}-\epsilon\right)$-approximation guarantee. Borodin et
al.~\cite{10.1145/3086464, Borodin:2012:MDM:2213556.2213580} study a
bi-criteria optimization problem formulated as the sum of a submodular function
and the Max-Sum diversification objective under matroid constraints. They show
that the local search approach preserves the
$\left(\tfrac{1}{2}-\epsilon\right)$-approximation guarantee. In an effort to
make the state-of-the-art local search algorithms more efficient, Ceccarello et
al.~\cite{Ceccarello:2018:FCD:3159652.3159719} propose algorithmic approaches
for constructing core-sets with strong approximation guarantees, resulting in
efficient algorithms with comparable quality to the best known local search
algorithms~\cite{Abbassi:2013:DMU:2487575.2487636, 10.1145/3086464,
Borodin:2012:MDM:2213556.2213580}. A core-set is a small subset of the original
data set that contains an $\alpha$-approximate solution for the Max-Sum
diversification problem. Cevallos et al.~\cite{10.5555/3039686.3039695} extend
the local search approach to distances of negative type and design algorithms
with $O\left(1- \tfrac{1}{k}\right)$-approximation and $O(nk^2 \log k)$ running
time.

\smallskip

\noindent 
\textbf{Fair $k$-center clustering.}
In the $k$-center clustering problem the objective is to select $k$ centers such that the maximum distance of any point from its closest cluster center is minimized. Intuitively, cluster centers tend to be distributed in a way that optimizes data coverage. Thus, $k$-center clustering can serve as another mechanism to perform diverse data selection, albeit the optimization objective is different from Max-Min. Max-Min diversification and $k$-center clustering are closely related. In fact, the approximation algorithms by Gonzalez~\cite{Gonzalez1985ClusteringTM} for the clustering problem and by Ravi et al.~\cite{1994:HSC:2753204.2753214} and Tamir~\cite{Tamir} for Max-Min diversification, are all based on the same farthest-first traversal heuristic, and they all provide a $\tfrac{1}{2}$-approximation guarantee. Nonetheless, the analysis of the two algorithms is substantially different and it is not always the case that an algorithm for one problem is applicable to the other.

In recent work, Kleindessner et al.~\cite{pmlr-v97-kleindessner19a} introduced
the fair variant of the problem, where the centers are partitioned into $m$
different groups and the constraint of selecting $k_i$ elements per group is
enforced in the output of the process. It is easy to find examples where no optimal solution for the fair $k$-center problem is optimal for the Max-Min objective and vice versa (see example in Figure~\ref{fig:kCenter}). Furthermore, we note that an optimal solution for fair $k$-center clustering can be arbitrarily bad for the Max-Min objective (e.g.,  Figure~\ref{fig:kCenter2}). Consequently, the two
problems need to be studied independently. The fair k-center clustering problem can also be expressed by a partition matroid, for which Chen et al.~\cite{Chen2016} provide a 3-approximation algorithm with a quadratic runtime. Kleindessner et al.~\cite{pmlr-v97-kleindessner19a} provide a linear-time algorithm with a $\left(3\cdot 2^{m-1}-1\right)$-approximation, while more recent work improved this bound to $3(1+\epsilon)$~\cite{chiplunkar2020solve}, and $3$-approximation~\cite{Jones2020FairKV}. In our Appendix, by adapting the ideas for fair Max-Min diversification, we design a linear-time algorithm for fair $k$-center clustering that also achieves a constant 3-factor approximation.

\begin{figure}[t]
	\begin{subfigure}[t]{0.4\textwidth}
		\centering
		\includegraphics[width=\textwidth]{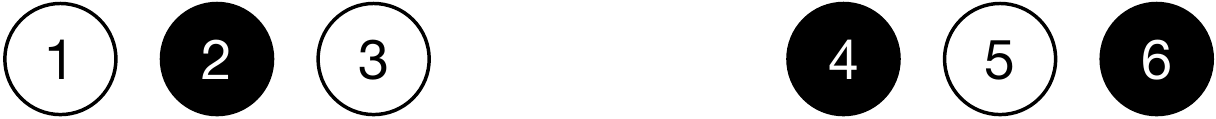}
		\captionsetup{justification=centering} 
		\caption{Example 1}
		\label{fig:kCenter}
		\vspace{-2mm}
	\end{subfigure}
	\hfill
	\begin{subfigure}[t]{0.4\textwidth}
		\centering
		\includegraphics[width=\textwidth]{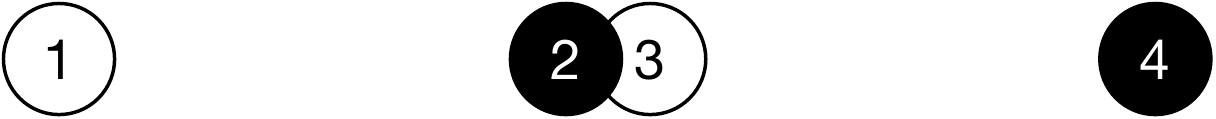} 
		\captionsetup{justification=centering}
		\caption{Example 2}
		\label{fig:kCenter2}
		\vspace{-2mm}
	\end{subfigure}
	\caption{(a)~An example where no optimal solution for the clustering
	problem is optimal for the diversity problem and vice versa. Suppose we have 	to	pick one white point and one black point. The unique optimal solution for
	clustering is $\{2,5\}$ whereas the unique optimal solution for Max-Min
	diversity is $\{1,6\}$. (b)~An optimal solution for the clustering problem 	may be arbitrarily bad for the diversity problem. Suppose we have to pick 	one white and one black point. Set $\{2,3\}$ is an optimal solution for 	clustering but yields an arbitrarily bad approximation ratio for the 	diversity problem as points $2$ and $3$ can be arbitrarily close together.
}
	\vspace{-1mm}	
\end{figure}	

\smallskip

\noindent
\textbf{Outline of contributions: Fair Max-Min diversification.}
To the best of our knowledge, this paper is the first to introduce fairness
constraints to Max-Min diversification. We initially focus on the case of disjoint groups, but extend our algorithms to tackle the overlapping case as well. Our work makes the following contributions.
\begin{itemize}[wide, labelwidth=!,labelindent=0pt,leftmargin=\parindent, itemsep=1pt, topsep=1pt]
	\item After some background and preliminaries (Section~\ref{sec:2.1}), we 	introduce and formally define the problem of \emph{fair} Max-Min 	diversification focusing on non-overlapping groups, and further discuss its 	complexity and approximability. To the best of our
	knowledge, no prior work has studied the Max-Min diversification objective
	under fairness constraints. We also describe how our
	algorithmic frameworks support any constraints that can be expressed in
	terms of partition matroids (Section~\ref{sec:2.2}).

	\item We propose a swap-based greedy approximation algorithm, with linear
	runtime, for the case of $m=2$, which offers a constant
	$\tfrac{1}{4}$-factor approximation guarantee (Section~\ref{sec:swapAlg}).

	\item We propose a general max-flow-based polynomial algorithm, with
	runtime linear in the size of the data, that offers a
	$\tfrac{1}{3m-1}$-factor approximation (Section~\ref{sec:flowAlg}). We also
	demonstrate that for constant $m$ and \emph{small} values for $k =o(\log
	n)$, we can achieve a constant $\tfrac{1}{5}$-approximation, also in linear
	time. While this bound is obviously stronger than our bound for the general
	case, the $\tfrac{1}{5}$-approximation algorithm becomes impractical as $k$
	increases (Section~\ref{sec:special}).
	 
	\item We generalize the \emph{fair} diversification problem to the case of
	overlapping groups (an element can belong to multiple demographic groups).
	We propose polynomial-time algorithms with $\tfrac{1}{4}$-factor 	approximation for the case of $m=2$ and $\tfrac{1}{3 {m \choose \lfloor
	m/2 \rfloor}-1}$-factor approximation for any $m$ 	(Section~\ref{sec:overlap}). 
\end{itemize}

	\section{Fair Max-Min Diversification: Background and Definition} \label{sec:Background} 
In this section, we review necessary background and preliminaries on the Max-Min diversification objective and relevant approximations. Then, we formally define the \emph{fair} Max-Min diversification problem, which generalizes Max-Min diversification. We further characterize the hardness of the problem and the hardness of its approximation, and describe its connection to partition matroids.

\subsection{Max-Min Diversification}\label{sec:2.1} 
\noindent
\textbf{Problem definition.} Prior work has identified a range of diversity
objectives to perform diverse data selection. In this work we primarily focus
on the Max-Min objective, which corresponds to the minimum distance of any two
items in a set $\mathcal{S}$. More formally, we assume a universe of elements $\cU$ and a pseudometric distance function $d: \mathcal{U} \times
\mathcal{U}\rightarrow \mathbb{R}_{0}^{+}$. For every $u, v \in \mathcal{U}$, $d$ satisfies the following properties: $d(u,u) = 0$, 
$d(u, v)=d(v,u)$ (symmetry),
and $d(u, v) \leq d(u, w) + d(w,v)$ (triangle inequality). Then, $d(u, v)$ captures the dissimilarity of the elements $u,v\in
\mathcal{U}$, and the Max-Min diversity score of a set $\mathcal{S}$ is
\[ \diver(\mathcal{S}) = \min\limits_{\substack{u,v \in \mathcal{S}\\ u\neq v}}d(u,v) \] Max-Min diversification seeks to identify a set $\mathcal{S}\subseteq\mathcal{U}$ and $|\mathcal{S}|=k$, such that the minimum
pairwise distance, $\diver(\mathcal{S})$, of elements in $\mathcal{S}$ is
maximized. 

\smallskip

\noindent
\textbf{Algorithms and approximations.} 
This problem formulation was initially studied in the operation research literature by Ravi et al.~\cite{1994:HSC:2753204.2753214} and in the context of facility location on graphs by Tamir~\cite{Tamir}. They both show that the problem is NP-complete even in metric spaces and give a greedy algorithm, GMM, that guarantees a $\tfrac{1}{2}$-approximation for Max-Min diversification. Ravi et al.~\cite{1994:HSC:2753204.2753214} also show that this problem cannot
be approximated within a factor better than $\tfrac{1}{2}$ unless P=NP through
a reduction from the clique problem. 

The GMM approximation algorithm uses the simple and intuitive farthest first traversal heuristic: Given a set of items $\mathcal{S}$, add the element from $\mathcal{U}$ whose minimum distance from any element in $\mathcal{S}$ is the largest. Algorithm~\ref{algo:GMMA} shows the pseudocode for \gonzalez, which starts with an initial set of elements $I$ and greedily augments it with $k$ elements from $\mathcal{U}$. Note that the
\gonzalez algorithm, as presented by Ravi et al.~\cite{1994:HSC:2753204.2753214} and Tamir~\cite{Tamir} assumes that
$I=\emptyset$; in this paper, we use the slight variant presented
in Algorithm~\ref{algo:GMMA}, which assumes that $I$ can be
non-empty. We use \gonzalez as a building block for the algorithms we present
in this paper. A naive implementation of the algorithm requires $O((|I|+k)^2 n )$ time but more efficient implementation requires $O((|I|+k) n )$ time; see, e.g.,~\cite{pmlr-v97-kleindessner19a, 10.14778/3192965.3192969} for details. 

\subsection{Fair Max-Min Diversification}\label{sec:2.2} 
\noindent
\textbf{Problem definition and analysis.}
We assume a universe of elements $\mathcal{U}$ of size $n$, comprising of $m$
non-overlapping classes: $\mathcal{U}=\bigcup_{i=1}^m\mathcal{U}_i$; we further
assume a pseudometric distance function $d: \mathcal{U} \times
\mathcal{U}\rightarrow \mathbb{R}_{0}^{+}$; finally, we assume non-negative
integers $\langle k_1,\dots,k_m\rangle$, which we call \emph{fairness
constraints}. Our goal is to identify a set $\mathcal{S}\subseteq\mathcal{U}$,
such that for all $i$, $|\mathcal{S}\cap\mathcal{U}_i|=k_i$, and the minimum
distance of any two items in $\mathcal{S}$ is maximized.  More formally:
\vspace{-2mm}
\begin{flalign*}
	\fMM: \ & \maximize_{\mathcal{S} \subseteq \mathcal{U}} \quad \min_{\substack{u,v\in 	\mathcal{S}\\ u\neq v}}d(u,v) & \\
	& \textup{subject to} \ |\mathcal{S} \cap \mathcal{U}_i| = k_i, \ \forall i \in [m]
\end{flalign*}

\looseness -1
Intuitively, \fMM aims to derive the set with the maximum diversity score
$\diver(\mathcal{S})$, while satisfying the fairness constraints. Next, we
state formally the hardness of \fMM and bound its approximability. These
results follow easily from the corresponding prior results on unconstrained
Max-Min diversification, as that problem reduces to \fMM for $m=1$.

\begin{corollary}[\textsc{Hardness}]\label{col:hard}
	Determining if there exists a solution to \fMM
	with diversity score $\geq\delta$ is NP-complete.
\end{corollary}

\begin{proof}	
	The problem is clearly in NP: If we are given a solution $\mathcal{S}$, we 	can verify that it satisfies the fairness constraints and compute its 	diversity score in polynomial time. The unconstrained version of Max-Min 	diversification is NP-complete~\cite{1994:HSC:2753204.2753214, Tamir}, and 	it is a special case of our problem for $m=1$. Since any instance of Max-Min 	diversification can be reduced to an instance of \fMM with $m=1$, then \fMM 	is also NP-complete.
\end{proof}

\begin{corollary}[\textsc{Approximability Bound}] \label{col:bound}
	There exists no polynomial-time $\alpha$-approximation algorithm for \fMM 	with $\alpha>\tfrac{1}{2}$, unless P=NP. 
\end{corollary}

\begin{proof}
	Suppose that there exists a polynomial algorithm that approximates the 	diversity score of the optimal solution to \fMM by a factor of 	$\alpha>\frac{1}{2}$. Then, this algorithm could also solve the unconstrained 	Max-Min diversification problem with approximation factor $\alpha$. However, 	Ravi et al.~\cite{1994:HSC:2753204.2753214} have shown that unconstrained 	Max-Min diversification cannot be approximated within a factor better than 	$\tfrac{1}{2}$, through a reduction from the clique problem. Therefore, it is 	not possible for such an algorithm to exist.
\end{proof}	

\subsubsection{Fairness as a Partition Matroid}\label{sec:fairnessAsPM} 
\looseness -1
While the focus of our work is on fairness constraints in particular, our results apply in general to any type of constraints that can be expressed in terms of a partition matroid. We provide a brief overview of the matroid definition and show that fairness constraints can be expressed as a partition matroid.

\begin{definition}
A matroid $\mathcal{M}$ is a pair $(\mathcal{E},
\mathcal{I})$ where $\mathcal{E}$ is a ground set of elements and $\mathcal{I}$ is a collection of subsets of $E$ (called \textit{independent sets}). All the \textit{independent sets} in $\mathcal{I}$ satisfy the following properties:

\begin{itemize} 
	\item If $\mathcal{A} \in \mathcal{I}$, then for every subset $\mathcal{B} 	\subseteq \mathcal{A}$,  $\mathcal{B}\in
	\mathcal{I}$. (Hereditary property) 
	\item If $\mathcal{A}, \mathcal{B} \in
	\mathcal{I}$ with $|\mathcal{A}| > |\mathcal{B}|$, then $\exists e \in 	\mathcal{A} \setminus \mathcal{B}$
	such that $\mathcal{B} \cup \{e\} \in \mathcal{I}$. (Exchange
	property) 
\end{itemize}

\end{definition}

A maximal independent set in $\mathcal{I}$ (also called a basis for a matroid)
is a set for which there is no element outside of the set that can be added so
that the set remains independent. All maximal independent sets of a
matroid have equal cardinality which is also called the rank of the matroid, rank($\mathcal{M})$. 

\begin{definition} 
    A matroid $\mathcal{M}=(\mathcal{E}, \mathcal{I})$ is a partition matroid
    if $\mathcal{E}$ can be decomposed into $m$ \textit{disjoint} sets
    $\mathcal{E}_1,\mathcal{E}_2, ..., \mathcal{E}_m$ and $\mathcal{I}$ is
    defined as $ \mathcal{I}= \{ S\subseteq \mathcal{E}: |S \cap \mathcal{E}_i| \leq k_i \
    \forall \ i \in [1, m] \}$.
\end{definition}

Note that a maximal independent set (or a basis) for a partition matroid is an
independent set that satisfies all the cardinality constraints with equality. 
For further information on matroids, we refer the interested reader to~\cite{schrijver2003combinatorial}. Based on the definitions above, in \fMM
the ground set is the universe of elements $\mathcal{U} = \bigcup_{i=1}^{m}
\mathcal{U}_i$. Then \fMM can be expressed as searching for the maximal independent set of the partition matroid defined over $\mathcal{U}$ that maximizes the Max-Min diversity function. 

\medskip

\noindent
\textbf{Our contributions to fair Max-Min diversification.} 
To the best our knowledge, this is the first paper to augment the Max-Min diversification problem with fairness constraints. For this problem, typically $m$ is a small constant and $k\ll n$. Therefore, when considering algorithmic complexity, we want to avoid high-order dependence on the size of the data, $n$.
In Section~\ref{sec:divAlgs}, we provide linear-time algorithms, with respect to $n$, with strong approximation guarantees for this problem in the case of non-overlapping groups. In Section~\ref{sec:overlap}, we extend our results to design polynomial-time algorithms with strong approximation guarantees for the
generalized setting of overlapping groups.
	\section{Approximating Diversity}\label{sec:divAlgs}
In Section~\ref{sec:2.2}, we showed that the \emph{fair} formulation for the
Max-Min diversification problem is NP-hard, and cannot be approximated within a
factor better than $\tfrac{1}{2}$. In this section, we propose three
approximation algorithms for this problem, with a best overall bound of
$\tfrac{1}{4}$ for the case of $m=2$. For ease of exposition, in the rest of
the paper we frequently refer to each of the $m$ groups as different colors.

Our algorithms use \gonzalez (Algorithm~\ref{algo:GMMA}) as a building block,
but adapting \gonzalez for fair Max-Min diversification is not straightforward.
We give an example of a simple and intuitive algorithm based on \gonzalez that
can lead to an arbitrarily bad result, even in the case of $m=2$ colors. In the
first phase of the algorithm, we use \gonzalez to greedily select elements of
any color until the constraints for one of them are satisfied. In the second
phase, we allow \gonzalez to greedily select the remaining elements only from
the under-satisfied color. Suppose that our data consist of one white and three
black elements positioned in a line as follows:

\begin{minipage}[h]{\textwidth}
	\vspace{1mm}
	\centering
	\includegraphics[width=0.4\textwidth]{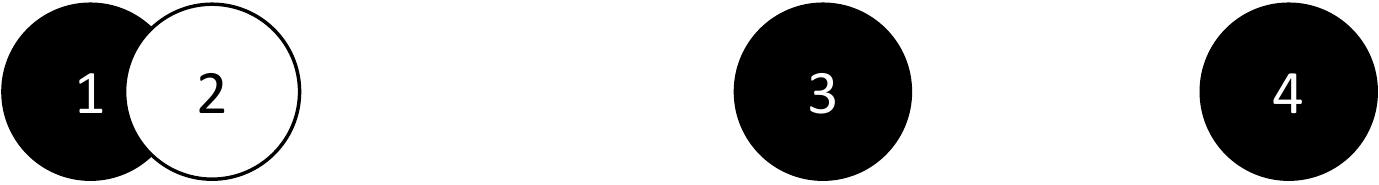}
	\vspace{1mm}
\end{minipage}

Further, consider that the fairness constraints require the selection of one
white and two black elements, and that \gonzalez first selects a black element.
Regardless of which black element is selected first, the simple algorithm we
described will always be forced to select elements 1 and 2---the possible
selection scenarios are: $\{1, 4, 2\}$, $\{3, 1, 2\}$, and $\{4, 1,
2\}$---which can be arbitrarily close to one another. This example demonstrates
how the choices made for one color, can lead to arbitrarily bad choices for the
other color(s), and the problem gets harder as $m$ increases.

Our algorithms employ \gonzalez in ways that guarantee the preservation of
\emph{good} choices for all colors. We start with a swap-based algorithm that
offers a $\tfrac{1}{4}$ approximation when $m=2$. Then we present a flow-based
algorithm with a $\tfrac{1}{3m-1}$ approximation when $m\geq 3$. Both
algorithms run in $O(kn)$ time. Finally, we present a
$\tfrac{1}{5}$-approximation for $m\geq 3$ that also runs in $O(kn)$, on the
assumption $m$ is constant and $k=o(\log n)$. However, the running time of this
third algorithm has an additional factor that depends exponentially on $k$,
which makes the algorithm practical only for \emph{small} $k$ values, e.g., for
$n=10^4$, $k \approx 10$.

\newcommand{\GMM}{
	\begin{algorithm}[H]
		\caption{\gonzalez Algorithm}\label{algo:GMMA}
		{\small
			\begin{algorithmic}[1] 
				\Statex
				\begin{description}
					\item[\rlap{Input:}\phantom{Output:}] $\mathcal{U}$: Universe of available elements 
					\item[\phantom{Output:}] $k \in \mathbb{Z}^{+}$ 
					\item[\phantom{Output:}] $I$: An initial set of elements
					\item[Output:] $\mathcal{S} \subseteq \mathcal{U}$ of size $k$ 
				\end{description}	
			
				\Procedure{GMM$(\mathcal{U},I, k)$}{}
				\State $\mathcal{S} \gets \emptyset$. 
			
				\If {$I = \emptyset$}
					\State $S \gets$ a randomly chosen point in $\mathcal{U}$ 
				\EndIf
				\While{$|\mathcal{S}|< k$} 
	    			\State $x \gets \underset{u \in \mathcal{U}}{\text{argmax}} \ \ 			   		\underset{s \in \mathcal{S}\cup I}{\text{min}} \ d(u, s)$ 
	    			\State $\mathcal{S} \gets \mathcal{S} \cup \{x\}$ 
				\EndWhile 
				\Statex
				\Return $\mathcal{S}$ 
				\EndProcedure 
			\end{algorithmic}
		}
	\end{algorithm}
}

\newcommand{\fairswap}{
	\begin{algorithm}[H]
		\caption{\fMMSwap: Fair Diversification for $m=2$} \label{algo:binary}
		{\small
		\begin{algorithmic}[1]
			\Statex 
			\begin{description}
				\item[\rlap{Input:}\phantom{Output:}]
				$\cU_1, \cU_2$: Set of points of color $1$ and $2$
				\item[\phantom{Output:}] $k_1, k_2  \in \mathbb{Z}^{+}$
				\item[Output:] $k_i$ points in $\cU_i$ for $i\in \{1,2\}$
			\end{description}

			\Procedure{\fMMSwap}{}
				\LineCommentx{Color-Blind Phase:}
				\State $\cS \leftarrow \gonzalez(\cU,\emptyset,k)$\label{ln:sGMM}
				\State $\cS_i=\cS \cap \cU_i$ for $i\in \{1,2\}$					
				\LineCommentx{Balancing Phase:}
				\State Set $U=\argmin_i (k_i-|\cS_i|)$ \Comment{{\footnotesize Under-satisfied set}}
				\State $O=3-U$ \Comment{{\footnotesize Over-satisfied set}}
				\State Compute: \label{ln:rebalance}
				\begin{eqnarray*}
					E & \leftarrow & \gonzalez(\cU_U,\cS_U,k_U-|\cS_U|)\\
		 			R &\leftarrow & \{\argmin_{x\in \cS_O} d(x,e): e\in E\}
		 		\end{eqnarray*}
				\Return $(\cS_U\cup E) \cup (\cS_O\setminus R)$
			\EndProcedure
		\end{algorithmic}
		}
	\end{algorithm}
}

\newcommand{\fairGMM}{
	\begin{algorithm}[H] 
		\caption{\fGMM: Fair Diversification for small $k$}\label{algo:smallk} 
		{\small
		\begin{algorithmic}[1]
			\Statex
			\begin{description} 
				\item[\rlap{Input:}\phantom{Output:}]
				$\cU_1, \ldots, \cU_m$: Universe of available elements
				\item[\phantom{Output:}] $k_1, \ldots, k_m \in \mathbb{Z}^{+}$
				\item[Output:] $k_i$ points in $\cU_i$ for $i\in [m]$ 			\end{description} 
			\Procedure{\fGMM}{}
				\For{$i\in [m]$}
					$Y_i\leftarrow \gonzalez(\cU_i,\emptyset, k)$
				\EndFor
				\State By exhaustive search, find the  sets $\cS_i \subseteq
				Y_i$ for $i\in [m]$ such that $|\cS_i|=k_i$ and
				$\diver(\cS_1 \cup \ldots \cup \cS_m)$ is maximized. 
			\EndProcedure
		\end{algorithmic}
		}
	\end{algorithm}
}

\newcommand{\fairFlow}{
	\begin{algorithm}[H] 
		\caption{\fMMFlow: Fair Diversification for $m\geq 3$}\label{algo:m3} 
		{\footnotesize
		\begin{algorithmic}[1]
			\Statex
			\begin{description} 
				\item[\rlap{Input:}\phantom{Output:}]
	$\cU_1, \ldots, \cU_m$: Universe of available elements 					\item[\phantom{Output:}] $k_1, \ldots, k_m \in \mathbb{Z}^{+}$
				\item[\phantom{Output:}] $\gamma\in \mathbb{R}$: A guess of the optimum fair diversity
				\item[Output:] $k_i$ points in $\cU_i$ for $i\in [m]$ 
				\end{description} 
				\Procedure{\fMMFlow}{}
				\For{$i\in [m]$}
					\State $Y_i\leftarrow \gonzalez(\cU_i,\emptyset, k)$
				\EndFor
				
				\State $Z_i \leftarrow$ maximal prefix of $Y_i$ such that all
				points 
				\Statex \hspace{\algorithmicindent}\phantom{$Z_i \leftarrow$}
				in $Z_i$ are $\geq d_1=\frac{m\gamma}{3m-1}$ apart.
				\State Construct undirected graph $G_Z$ with nodes   
				\Statex \hspace{\algorithmicindent}$Z=\bigcup_i
				Z_i$ and edges $(z_1,z_2)$, if $d(z_1,z_2)< 
				d_2=\frac{\gamma}{3m-1}$.
				\State $C_1, C_2, \ldots C_t\leftarrow$ Connected components of
				$G_Z$.
				\LineCommentx{Construct flow graph}
				\State Construct directed graph $G=(V,E)$ where\label{ln:reduction}
				\begin{eqnarray*}
					V&=& \{a,u_1, \ldots, u_m, v_1, \ldots, v_t, b\}  \\
					E&=& \{(a,u_i) \mbox{ with capacity $k_i$}: i \in [m]\} \\
					& & \cup~ \{(v_j,b) \mbox{ with capacity $1$}: j \in [t]\} \\
					& & \cup~  \{(u_i,v_j) \mbox{ with capacity $1$}:  |Z_i\cap
					C_j|\geq 1\} 
				\end{eqnarray*} 
				\State Compute max $a$-$b$ flow. 
				\If{flow size $<k=\sum_i k_i$}
					\Return $\emptyset$ \Comment{Abort}
				\Else \Comment{max flow is $k$}
					\State $\forall(u_i,v_j)$ with flow add a node in $C_j$ with color $i$ to $\mathcal{S}$.
				\EndIf
				\Statex	
				\Return $\mathcal{S}$
			\EndProcedure
		\end{algorithmic}}
	\end{algorithm}
}

\begin{figure}[t]
	
	\scalebox{0.82}{
	\begin{minipage}[t]{0.6\textwidth}
		\GMM
		\fairswap
	\end{minipage}
	}	
	~\scalebox{0.82}{	
	\begin{minipage}[t]{0.62\textwidth}
		\fairFlow \vspace{-4mm}
		\fairGMM
	\end{minipage}
	}
\end{figure}

\subsection[Fair-and-Diverse Selection: m=2]{Fair-and-Diverse Selection: {\large\boldmath{$m=2$}}}\label{sec:swapAlg} 
In the binary setting, the input is a set of points $\mathcal{U} = \mathcal{U}_{1} \cup \mathcal{U}_{2}$ and two non-negative integers $\langle k_1, k_2\rangle $ with $k_i \leq |\mathcal{U}_i|$ for all $i \in \{1,2\}$. We want to select a set $\mathcal{S}$ with $k_i$ elements from each $\mathcal{U}_i$ partition such that the $\diver(\mathcal{S})$ is maximized. 

\smallskip

\noindent
\textbf{Algorithm and intuition.} \fMMSwap (Algorithm~\ref{algo:binary}) has two phases; the \emph{color-blind} and the \emph{balancing} phase. In the \emph{color-blind} phase, we call \gonzalez by initializing $I$ to
the empty set so as to retrieve a set $\mathcal{S} = \mathcal{S}_1 \cup
\mathcal{S}_2$ of size $k$ (line~\ref{ln:sGMM}). If $|\mathcal{S}_1|=k_1$ and $|\mathcal{S}_2|=k_2$ then these two sets are returned. Alternatively, if one set is smaller than required, then the other set is larger than required, and we need to rebalance these sets.  

Let $\mathcal{S}_U$ be the set that is too small and let $\mathcal{S}_O$ be the set that is too large. The algorithm next  finds $k_U-|\mathcal{S}_U|$  extra points $E\subseteq \cU_U$ to add to $\mathcal{S}_U$ by again using the \gonzalez algorithm, this time initialized with the set $\mathcal{S}_U$. For each point in $E$ we then remove the closest point in  $\mathcal{S}_O$ (line~\ref{ln:rebalance}). In this way we add $k_U-|\mathcal{S}_U|$ points to $\cS_U$ and remove $k_U-|\mathcal{S}_U|$ points from $\cS_O$. After this rebalancing the size of $\cS_U$ is $$|\mathcal{S}_U|+(k_U-|\mathcal{S}_U|)=k_U$$ 
and the size of $\cS_O$ is $$|\mathcal{S}_O|-(k_U-|\mathcal{S}_U|)=k-k_U=k_O$$ as required. Note that sets $E$ and $R$ will be empty if the sets are already balanced after the color blind phase and thus the set $\mathcal{S}$ will not be altered by the balancing phase.

\smallskip

\noindent
\textbf{Running-time analysis.} The running time of \fMMSwap (Algorithm~\ref{algo:binary}) is $O(kn)$. In the color-blind phase of the algorithm we run \gonzalez on $\mathcal{U}$ with $I=\emptyset$ and this takes $O(kn)$ time. Then in the balancing phase, computing the extra points $E$ via the \gonzalez algorithm takes $O(kn)$ time and computing $R$ takes $O(k^2)$ time since there are fewer than $k$ points in $E$ and at most $k$ points in $\cS_O$.

\smallskip

\noindent
\textbf{Approximation-factor analysis.} Let $\cS^*$ be the set of $k$ points in $\cU$ that maximize the diversity when there are no fairness constraints. Let $\ell^*=\diver(\cS^*)$. Let $\cF^*=\cF^*_{1}\cup \cF^*_{2}$ be the set of $k$ points in $\cU$ that maximize the diversity subject to the constraint that for each $i\in \{1,2\}$, $k_i$ points are chosen of color $i$. Let $\ell^*_{\text{fair}}=\diver(\cF^*)$ and note  that $\ell^*\geq \ell^*_{\text{fair}}$. We first argue that 

\begin{equation} \label{eq:eq1}
\diver(\cS)\geq \ell^*/2\geq \ell^*_{\text{fair}}/2. 
\end{equation}

This follows because, by the triangle inequality, there is at most one point in $\cS^*$ that is distance $<\ell^*/2$ from each point in $\cS$; otherwise two points in $\cS^*$ would be $<\ell^*$ apart and this contradicts the fact $\diver(\cS^*)=\ell^*$. Hence, while the \gonzalez  algorithm has picked $<k$ elements, there exists at least one element in $\cS^*$ that can be selected that is distance  $\geq \ell^*/2$ from all the points already selected. Since the algorithm picks the next point farthest away from the points already chosen, the next point is at least $\ell^*/2$ from the existing points. Next we argue that

\begin{equation} \label{eq:eq2}
 \diver(\cS_U\cup E)\geq \ell^*_{\text{fair}}/2. 
\end{equation} 

To show this, first observe that, $\diver(\cS_U) \geq \diver(\cS) \geq \ell^*_{\text{fair}}/2$. Next consider the points added to $E$ by \gonzalez. By the triangle inequality there is at most one point in $\cF^*_U$ that is distance $<\ell^*_{\text{fair}}/2$ from each point in $\cS_U\cup E$. Hence, while \gonzalez has picked $<k-|\cS_U|$ elements, there exists at least one element that can be selected that is distance  $\geq \ell^*_{\text{fair}}/2$ from the points already selected. Since the algorithm picks the next point farthest away from the points already chosen, the next point is at least $\ell^*_{\text{fair}}/2$ from the existing points. 

From Eq.~\ref{eq:eq1} and Eq.~\ref{eq:eq2}, we can guarantee that $d(x,y)\geq \ell^*_{\text{fair}}/2$ for all  pairs of points $x,y\in \cS_U\cup E \cup \cS_O$ except potentially when $x \in E$ and $y\in \cS_O$. To handle this case, for each $x\in E$ we remove the closest point in $\cS_O$. Note that by an application of the triangle inequality and the fact that $\diver(\cS_O)\geq \lsf/2$,  for each $x\in E$ there can be at most one point $y\in \cS_O$ such that $d(x,y)< \lsf/4$. Hence, after the removal of the closest points the distance between all pairs is $\geq \lsf/4$ as required. We summarize the analysis of this section as follows:

\begin{theorem}
\fMMSwap (Algorithm~\ref{algo:binary}) is a  $1/4$-approximation algorithm for the fair diversification problem when $m=2$ that runs in time $O(kn)$. 
\end{theorem}

\noindent
\textbf{Connections to prior art.} 
Balancing mechanisms have also been successfully applied to matroid
optimization settings subject to fairness
constraints~\cite{pmlr-v89-chierichetti19a}, and to the red-blue matching
problem~\cite{nomikos}. However, our objective function cannot be expressed by
a matroid (or an intersection of matroids), and thus the results of prior work
are not applicable to our setting. Further, the algorithms and analysis are
also distinct for these problems.
	\subsection[Fair-and-Diverse Selection: m>=3]{Fair-and-Diverse Selection: {\large\boldmath{$m\geq 3$}}}\label{sec:flowAlg}

\smallskip

\noindent
\textbf{Basic algorithm.} 
We start by presenting a basic algorithm that takes as input a guess $\gamma$ for the optimum fair diversity. If this guess is greater than the optimum fair diversity then the algorithm may abort, but if the algorithm does not abort, it will return a fair diversity at least $\gamma/(3m-1)$.

\smallskip

\noindent
\textbf{Algorithm and intuition.} \looseness-1 
The approach of \fMMFlow (Algorithm~\ref{algo:m3}) is to construct disjoint sets of points $C_1, C_2, \ldots$ such that, if $\gamma$ is at most the optimal fair diversity, it is possible to find sets $\cS_1, \ldots, \cS_m$ of sizes $k_1, \ldots, k_m$ such that each $C_i$ contains at most one point from $\cS_1 \cup \ldots \cup \cS_m$. If we can construct $C_1, C_2, \ldots$ such that for any $x\in C_i$ and $y\in C_j$, then $d(x,y)\geq d_2$ for some value $d_2$ then we have $$\diver(\cS_1\cup  \ldots \cup \cS_m) \geq d_2$$ 

Furthermore, because the sets $C_1, C_2, \ldots$ are disjoint it is possible to find sets $\cS_1, \ldots, \cS_m$ with the required property via a reduction to network flow (noting that the optimal flow in a network with integer capacities is always integral). See the algorithm for the precise reduction and see Figure~\ref{fig:network} for an example.

\begin{figure}[t]
	\centering
	\includegraphics[width=0.4\textwidth]{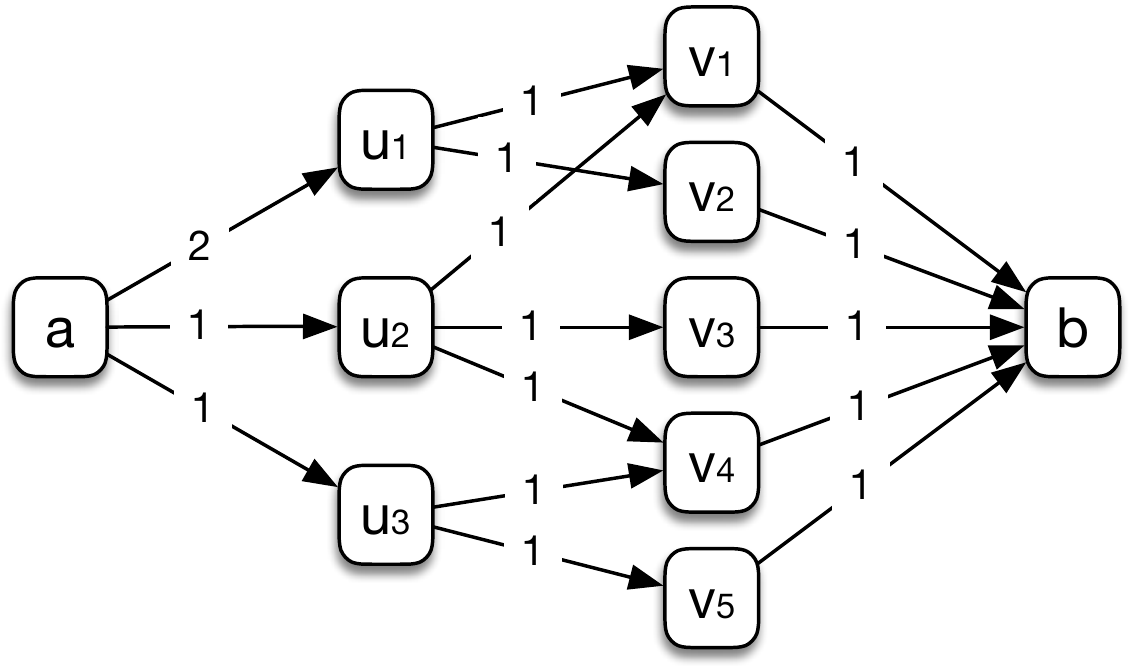}
	\vspace{-3mm}
	\caption{The graph construction in Algorithm~\ref{algo:m3} 	(line~\ref{ln:reduction}) corresponding to	$m=3$, $k_1=2, k_2=1, k_3=1$. 	Points of color 1 are contained in $C_1$ and $C_2$. Points of color 2 are 	contained in $C_1,C_3,$ and $C_4$. Points of color 3 are contained in $C_4$ 	and $C_5$. Note there is an $a$-$b$ flow of size $k_1+k_2+k_3$ iff it is 	possible to pick at most one point from each $C_j$ while still picking at 	most $k_i$ points of color $i$ for each $i\in [m]$.
	}
	\label{fig:network}
	\vspace{-3mm}
\end{figure}

The way we construct each $C_1, C_2, \ldots$ is to first run \gonzalez on each color $i$ and use this to identify at most $k$ points $Z_i$ of color $i$ such that $\diver(Z_i)\geq d_1$ for some value $d_1$ to be determined. We then partition $\bigcup_i Z_i$ into the disjoint groups $C_1, C_2, \ldots$ where the partition satisfies the property that any two points $x,y\in \bigcup_i Z_i$ such that $d(x,y)<d_2$ are in the same group. Note that $x,z$ will end up in the same group if there exists $y$ such that $d(x,y)<d_2$ and $d(y,z)<d_2$; more generally two points can end up in the same group because of a chain of points where each adjacent pair of points are close. However, in the analysis, we will show that these chains cannot be too long and, for appropriately chosen $d_1$ and $d_2$,  any two points in  $C_j$ are distance $<d_1$ from each other. In the analysis, this will enable us to argue that if $\gamma$ is at most the optimal fair diversity, it is possible to find the required sets $\cS_1, \ldots, \cS_m$.

\smallskip

\noindent
\textbf{Analysis of basic algorithm.} We need a preliminary  lemma that argues that all the points in the same connected component are close together.

\begin{lemma}\label{lem:groupwidth}
For all connected components $C_j$, 
\[\forall x,y \in C_j:~d(x,y) < (m-1)d_2,\]
and $C_j$ does not contain any two points of the same color.
\end{lemma}

\begin{proof} \looseness -1
Consider two points $x,y\in C_j$ and let the length of a shortest unweighted path $P_{x,y}$ between $x$ and $y$ in the graph be $\ell$. If $\ell\leq m-1$ then  $d(x,y)< (m-1)d_2$ as required. If $\ell\geq m$ then there exists two points on this path (including end points) that have the same color and this will lead to a contradiction. Consider the subpath $P_{x',y'}\subset P_{x,y}$ where $x'$ and $y'$ have the same color $i$ and all internal nodes have distinct colors. Then the length of $P_{x',y'}$ is strictly less than \[md_2=m\gamma /(3m-1)=d_1\] 
But this contradicts $d(x',y')\geq d_1$ for all points in $Y_i$. 
\end{proof}

The next theorem establishes that when the algorithm does not abort, the solution returned has diversity at least $\gamma/(3m-1)$ and that it never aborts if the guess $\gamma$ is at most the optimum diversity.

\begin{theorem} 
Let $\lsf$ be the optimum diversity. If $\gamma\leq \ell^*_{\text{fair}}$ then the algorithm returns a set of points of the required colors that are each $\geq \gamma/(3m-1)$ apart.  If $\gamma> \ell^*_{\text{fair}}$ then the algorithm either aborts or returns a set of points of the required colors that are each $\geq d_2=\gamma/(3m-1)$ apart.  
\end{theorem}

\begin{proof}
Note that if the algorithm does not abort  then all points are $\geq \gamma/(3m-1)$ apart since any two points in different connected components are $\geq \gamma/(3m-1)$ apart. 

Hence, it remains to argue that if $\gamma\leq \lsf$ then the algorithm does not abort. To argue this, we will construct a flow of size $k$ in the network instance. And to do this it suffices to identify $k_i$  connected components including a point from $Z_i$ for each $i$, such that the resulting set of $k_1+k_2+\ldots + k_m$ connected components are all distinct.  To do this, we start by defining a node $u_i$ to be critical if $|Z_i|<k$ and non-critical otherwise. Let $O_i\subset \cU_i$ be the set of $k_i$ points in the optimum solution. For $x\in \cU_i$, let $f(x)$ be the closest point $Z_i$ to $x$. If $u_i$ is critical, then note that $d(x,f(x))<d_1$. Note that for all points $x,y\in \cup_{i:\textup{critical}} f(O_i)$, \[d(x,y)> \lsf-2d_1\geq \gamma-2\gamma m/(3m-1)=(m-1)d_2\] and hence, by Lemma~\ref{lem:groupwidth}, this implies that all points in $\bigcup_{i:\textup{critical}} f(O_i)$
are in different connected components. We then consider each  non-critical node $u_i$ in turn. Since $u_i$ was non-critical and each connected component has at most one point in each $Z_i$, there are $k$ connected components that include a point in $Z_i$. At most $k -k_i$ need to be used to pick points of other classes and hence at least $k-(k -k_i)=k_i$ remain.
\end{proof}

\smallskip

\noindent
\textbf{Final algorithm.} \looseness -1
Our final algorithm is based on binary searching for a ``good'' guess $\gamma$ for the optimum diversity $\lsf$ where each guess can be evaluated using the basic algorithm  above. The goal is to find a guess that is close to $\lsf$ or larger such that the algorithm does not abort. There are two natural ways to do this; which is best depends on parameters of the data set. 

\smallskip

\noindent
\emph{Binary-searching over continuous range:}
For the first approach, note that $\lsf\in [d_{min},d_{max}]$ where
\[
d_{\min} = \min_{x,y\in X:x\neq y} d(x,y)  ~~~\mbox{ and }~~~
d_{\max} = \max_{x,y\in X} d(x,y) \ .
\]
Hence, there exists a guess $\gamma =(1+\epsilon)^i d_{\min}$ for some \[i\in \{0, 1,2 ,\ldots ,\lceil \log_{1+\epsilon} R \rceil\} \mbox{ where } R:=d_{\max}/d_{\min}\] such that $\lsf/(1+\epsilon) \leq \gamma \leq \lsf$. 
Note that for this guess, the algorithm returns a $(3m-1)(1+\epsilon)$ approximation. We can find this guess (or an even better guess, i.e., a $\gamma>\lsf$ for which the algorithm does not abort) via a binary search over the $1+\lceil \log_{1+\epsilon} R \rceil$ possible guesses. The number of trials  required is 
\[
O(\log( 1+\lceil \log_{1+\epsilon} R \rceil)=O(\log (\epsilon^{-1})+  \log \log  R) 
\]  

\smallskip

\noindent
\emph{Binary-searching over discrete set:} 
For the second approach we note that after the algorithm's initial step (which did not depend on the guess $\gamma$) there are only $km$ points and hence at most $\binom{km}{2}$ distinct distances between remaining points. Hence, it suffices to only consider guesses $\gamma$ such that $d_1$ or $d_2$ corresponds to one of these $O(k^2 m^2)$ values.
We can sort these values in $O(k^2 m^2 \log km)$ time and then binary search over this range to find a good guess using $O(\log km)$ trials.

\medskip

\noindent
\textbf{Final diversification result.} 
Our main theorem of this section follows by combining the binary search over a discrete set approach with the basic algorithm.

\begin{theorem}
There is a  $\tfrac{1}{3m-1}$-approximation algorithm for the fair diversity problem that  runs in time $O(kn + k^{2}m^{2} \log (km))$. 
\end{theorem}
\begin{proof}
The time  to construct $Y_1, \ldots, Y_m$ is $O(kn)$. We then need to sort the  $O(k^2 m^2)$ distances amongst these points. This takes $O(k^2 m^2 \log (km))$ time. The time to construct and solve the flow instance is $O(k^{2}m^{2})$ since the flow instance has $O(km)$ nodes and $O(km)$ edges~\cite{Orlin13, orlin}. Note that the binary search requires us to construct and solve $O(\log (km))$ flow instances. Hence the total running time is as claimed.
\end{proof}

If we used the binary search over a continuous range approach, the running time would by $O(kn+k^2 m^2 (\log \epsilon^{-1} + \log \log d_{\max}/d_{\min}))$ and the approximation ratio would be $\tfrac{1}{(3m-1)(1+\epsilon)}$.

\subsection[Fair-and-Diverse Selection: Small k, m]{Fair-and-Diverse Selection: Small {\large\boldmath{$k$}}, {\large\boldmath{$m$}}}\label{sec:special}
In this section, we present a simple algorithm that has the advantage of achieving a better approximation ratio than the algorithm in the previous section. The downside of the algorithm is that the running time is exponential in $k$, specifically, $O(kn + k^2(em)^k)$. However, when $m=O(1)$ and $k=o(\log n)$ the dominating term in the running time is $O(kn)$, as in the case of the algorithms from the previous sections.

\smallskip

\noindent
\textbf{Algorithm and intuition.} The basic approach of \fGMM (Algorithm~\ref{algo:smallk}) is to first select $k$ points (or less if there are fewer than $k$ points of a particular color) of each color via the \gonzalez algorithm. The resulting subset $\bigcup_i Y_i$ has at most $km$ points and this is significantly smaller than the original set of points assuming $k$ and $m$ are much smaller than $n$. Hence, it is feasible to solve the problem via exhaustive search on the subset of points. In the analysis, we will be able to show that the optimal fair diversity amongst the subset of points is at least $1/5$ of the optimal fair diversity amongst $\bigcup_i \cU_i$.

\smallskip

\noindent
\textbf{Analysis.} To prove the approximate factor we need to show that the optimal solution amongst the subset of points selected in step one has diversity that is not significantly smaller than the optimal diversity of the original set of points. To show this the basic idea is that for each $i$, the set $Y_i$ will contain at least one  point near every color $i$ point  in the optimal solution or will contain $k$ points such  that even if we remove any set of $k-k_i$ points to make space for points of other colors, the remaining set of $k_i$ points of color $i$ still has  sufficiently high diversity.

\begin{theorem}
\fGMM (Algorithm~\ref{algo:smallk}) returns a $\tfrac{1}{5}$-approximation and the running time is $O(kn + k^2(em)^k)$. Note that this is $O(kn)$ when $k=o(\log n)$ and $m=O(1)$.
\end{theorem}
\begin{proof}
	
For the running time, note that Step 1 can be implemented in $O(kn)$ time. For Step 2, note that there are at most $km$ points in $Y_1, Y_2, \ldots Y_m$ so a brute force algorithm needs to consider at most $\binom{km}{ k}\leq  (em)^k$ sets of points and computing the min distance for each takes $O(k^2)$ time. Note that this is $o(n)$ assuming $k=o(\log n)$ and $m$ is constant. 

For the approximation ratio, it suffices to argue that if $\lsf$ is the optimum value then there exists a set of points amongst $Y_1\cup \ldots \cup Y_m$ with the required colors that are $\lsf/5$ apart.  Let $Z_i$ be the maximal prefix of $Y_i$ such that all points at points are $\geq  2\lsf/5$ apart. For each $x \in \cU_i$, let $f(x)$ be the closest point in $Z_i$. Call $i$ critical if $|Z_i|<k$. Note that if $i$ is critical, then  $d(x,f(x))< 2\lsf/5$. Let $O_i$ be the optimal set of color $i$ points and consider the subsets $\cS_1, \cS_2, \ldots \cS_m$ of points in $Z_1, Z_2, \ldots Z_m$ defined as follows:

\begin{itemize}[wide, labelwidth=!,labelindent=0pt,leftmargin=\parindent, itemsep=1pt, topsep=1pt]
\item For all $i$ that are critical, let $\cS_i = f(O_i)$ and let $D=\cup_{i:\textup{critical}} \cS_i$. Note that $\diver(D)> \lsf-4\lsf/5=\lsf/5$.  

\item For each $j$ that is not critical: Remove all points in $Z_j$ that are distance $< \lsf/5$ from a point in $D$. Note that at most one point in $Z_j$ is $< \lsf/5$ from each point in $D$ because points in $Z_j$ are $\geq 2\lsf/5$ apart. Hence, at most $|D|$ points are removed from $Z_j$.

\item Process  the non-critical $j$ in arbitrary order: Pick $k_j$ points $\cS_j$ arbitrarily from $Z_j$. Remove all points from $Z$ that are distance $< \lsf/5$ from a point in $\cS_j$. This removes at most $k_j$ points from each $Z_i$. Note that when we process $j$ there are at least $k-(\sum_{i:\cS_i~\textup{defined so far}} k_i)\geq k-(k-k_j)=k_j$ points in $Z_j$.
\end{itemize}
Note $\diver(\bigcup_i \cS_i)\geq \lsf/5$ and this implies the claimed approximation factor.
\end{proof}
	\newcommand{\Genswap}{
	\begin{algorithm}[H]
		\caption{\fMMGSwap: Overlapping classes for $m=2$} \label{algo:swap2}
		{\small
			\begin{algorithmic}[1]
				\Statex 
				\begin{description}
					\item[\rlap{Input:}\phantom{Output:}] $\cU_1, \cU_2$: Universe of available elements
					\item[\phantom{Output:}] $\gamma \in \mathbb{R}$: A guess on the optimum fair diversity
					\item[\phantom{Output:}] $k_1, k_2  \in \mathbb{Z}^{+}$
					\item[Output:] at least $k_i$ points in $\cU_i$ for $i\in \{1,2\}$
    			\end{description}
				\Procedure{\fMMGSwap}{}
					\State $\cS_{\{1, 2\}} \leftarrow$ maximal subset of $X_{\{1, 					2\}}$ with all points $\geq \gamma/ 4$ apart \label{line:set1}
					
					\State $\setB \leftarrow$ all the points in $\cU$ that $< 					\gamma/4$ apart from a point in $\cS_{\{1, 2\}}$ \label{ln:close} 					
					\State $\setA \leftarrow \cU \setminus \setB$ 					\Comment{$\setA = \setA_{\{1\}} \cup \setA_{\{2\}} \subseteq X_{\{1\}} \cup X_{\{2\}}$}
					\Statex 
					\LineCommentx{Select the missing points to satisfy the constraints:}
					\State Set $t = |\cS_{\{1, 2\}}|$

					\If{$|\setA \cap \cU_i| \geq k_i-t$ for $i \in \{1, 2\}$} 
						\State $\cS_{\{1\}} \cup \cS_{\{2\}}\leftarrow $\Call{\fMMSwap}{$\setA, k_1-t, k_2 - t$} 
						\State $\cS \leftarrow \cS_{\{1\}} \cup \cS_{\{2\}} \cup \cS_{\{1, 2\}}$ 
						\Else \State $\cS \leftarrow \emptyset$ \Comment{Abort}
					\EndIf 
					\Statex
					\Return $\cS$ 
				\EndProcedure
			\end{algorithmic}
		}
	\end{algorithm}
}

\newcommand{\GenFlow}{
	\begin{algorithm}[H]
		\caption{\fMMGFlow: Overlapping classes for $m\geq3$}\label{algo:general} 
		{\small
		\begin{algorithmic}[1]
			\Statex
			\begin{description} 
				\item[\rlap{Input:}\phantom{Output:}] $\cU_1, \ldots, \cU_m$: Universe of available elements \smallskip
				\item[\phantom{Output:}] $c_L\in \mathbb{Z}^{+}$ for all $L\subset [m]$: A guess of the flow distribution \smallskip
				\item[\phantom{Output:}] $\gamma\in \mathbb{R}$: A guess of the optimum fair diversity \smallskip
				\item[\phantom{Output:}] $k_1, \ldots, k_{m} \in \mathbb{Z}^{+}$ 
				\item[Output:] at least $k_i$ points in $\cU_i$ for $i\in [m]$
				\end{description} 
			\Procedure{\fMMGFlow}{}
				\State Define \[d_1\leftarrow \frac{\comb \gamma}{3\comb-1} ~~~\mbox{ and }~~~ d_2\leftarrow \frac{\gamma}{3\comb-1}\]
				\State $Z_{[m]} \leftarrow $ maximal subset of $X_{[m]}$ with all points $\geq d_1$ apart \label{ln:Z}
				\For{$t=m-1, m-2, \ldots, 1$}
					\For{all sets of $L$ of size $t$}\label{ln:loop}
					\State $Z_{L} \leftarrow$ maximal subset of $X_{L}$ such that all points in $Z_L$ \Statex \hspace{.8in} are $\geq d_1$ from every other element in $$Z_{L} \cup \bigcup_{L'\in [m]: |L'|\geq t+1, L\subset L'} Z_{L'}$$
					\EndFor
				\EndFor
				\State Construct undirected graph $G_Z$ with nodes $Z=\bigcup_{L\subset [m]} Z_L$ 
				\Statex \hspace{\algorithmicindent}and edges $(z_1,z_2)$ if $d(z_1,z_2)< d_2$
				\State $C_1, C_2, \ldots C_t\leftarrow$ Connected components of $G_Z$
				\LineCommentx{Construct flow graph}
				\State Construct directed graph $G=(V,E)$ where
				\begin{eqnarray*}
					V&=& \{a, v_1, \ldots, v_t, b\} \cup \bigcup_{L\subset [m]: |L|>0}\{u_L\} \\
					E&=& \{(a,u_L) \mbox{ with capacity $c_L$}: \mbox{non-empty} \ L\subset [m] \} \\
					& & \cup~ \{(v_j,b) \mbox{ with capacity $1$}: j \in [t]\} \\
					& & \cup~  \{(u_L,v_j) \mbox{ with capacity $1$}:  |Z_L \cap C_j|\geq 1\} 
				\end{eqnarray*} 
				
				\State Compute max $a$-$b$ flow. 
				\If{flow size $<\sum_{L\subset [m]} c_L$}
					\Return $\emptyset$ \Comment{Abort}
				\Else 
					\State $\forall(u_L,v_j)$ with flow add a node in $C_j\in (\cap_{i\in L} \cU_i)$ to $\mathcal{S}$
				\EndIf
				\Return $\mathcal{S}$
			\EndProcedure
		\end{algorithmic}
	}
	\end{algorithm}	
}

\section{\mbox{Generalizing to Overlapping Groups}}\label{sec:overlap}
In this section, we show how we can extend our algorithmic framework to allow the elements in the universe $\cU$ to belong to multiple classes, e.g., an individual may belong to multiple demographic groups such as multiple races, or combinations of race, gender, and other sensitive demographics. First, we formally define the problem and show how our \fMMSwap and \fMMFlow algorithms can be adapted to support this generalized setting. 

We assume a universe of elements $\cU$ comprising of $m$ \emph{possibly} overlapping classes $\cU_1, \cU_2, \dots,  \cU_m$, a pseudometric distance function $d: \mathcal{U} \times \mathcal{U}\rightarrow \mathbb{R}_{0}^{+}$ and a set of fairness constraints $\langle k_1,\dots,k_m\rangle$ where each $k_i$ is a non-negative integer with $k_i \leq |\cU_i|$. Our goal is to identify a set $\mathcal{S}\subseteq\mathcal{U}$ to satisfy the fairness constraints such that the minimum distance of any two items in $\mathcal{S}$ is maximized. 

It will be convenient to introduce some additional notation. For any $L \subset [m]$, define \[X_L = \big (\bigcap_{i\in L} \cU_i \big ) \cap \big (\bigcup_{j\not \in L} \cU_j \big )\] 

That is, $X_L$ consists of all elements exactly in the classes of $L$ and no others. Note that if we select an element in $X_L$ it contributes to helping satisfy $|L|$ of the fairness constraints. Hence, it may be possible to satisfy all the constraints by picking fewer than $k_1+\ldots+k_m$ elements. Further, a feasible solution may require more than $k_i$ elements for class $i$ (example in Figure~\ref{fig:overlap}). Formally, we define the problem as follows:
\vspace{-2mm}
\begin{flalign*}
	\fMMG: \ & \maximize_{\mathcal{S} \subseteq \mathcal{U}} \quad \min_{\substack{u,v\in 	\mathcal{S}\\ u\neq v}}d(u,v) & \\
		& \textup{subject to} \ |\mathcal{S} \cap \mathcal{U}_i| \geq k_i, \ \forall i \in [m] 
\end{flalign*}

\begin{figure}[t]
	\centering
	\includegraphics[width=0.5\textwidth]{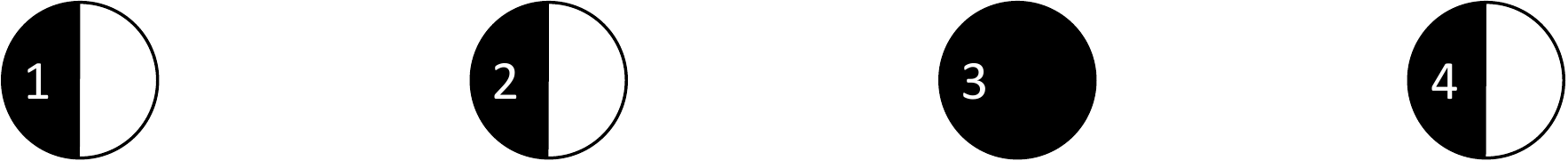}
	\caption{An example with $m=2$ overlapping classes, with $|\cU_1| = 3$ and 	$|\cU_2| = 4$, where (a)~the fairness constraints can be satisfied with 	fewer than $k$ elements and (b)~a class has to be overrepresented to satisfy 	the fairness constraints for all classes. Suppose we have to pick two white 	and one black element $(k=3)$. A feasible solution consists of two bi-colored 	elements, thus fewer than $k$, in which the black class is represented by two 	and not just one element.} \vspace{-3mm}
	\label{fig:overlap}
\end{figure}

\subsection{Fair-and-Diverse Selection (Overlaps): {\large\boldmath{$m=2$}}}
\label{sec:swapover}

In the binary setting, the input is a set of points $\cU$ that comprises of
$m=2$ overlapping classes; $\cU_1 = X_{\{1\}} \cup X_{\{1, 2\}}$ and $\cU_2 =
X_{\{2\}} \cup X_{\{1, 2\}}$. We design a swap-based algorithm, with
$1/4$-approximation guarantee, which uses the idea of binary searching over a
discrete set of guesses for the optimum fair diversity, denoted as $\lsf$.

\smallskip

\noindent
\textbf{Algorithm and intuition.} The \fMMGSwap algorithm (Algorithm~\ref{algo:swap2}) takes as input a guess $\gamma$ for the optimum fair diversity. We show that if $\gamma \leq \lsf$, we can always find enough points to construct a fair set $\cS = \cS_{\{1\}} \cup \cS_{\{2\}} \cup \cS_{\{1, 2\}}$ with $\diver{(\cS)} \geq \gamma/4$ (where $\cS_L=\cS\cap X_L$).

\looseness -1
The algorithm first finds as many points as possible in $X_{\{1,2\}}$ and are
at least $\tfrac{\gamma}{4}$ apart from each other. Let $\cS_{\{1, 2\}}$ be the
resulting set, with a total of $t$ points. Note that to satisfy the fairness
constraints, we need to add $k_i - t$ points for each class i in $\{1, 2\}$.
The algorithm proceeds to remove all points in $\cU$ that are closer than
$\tfrac{\gamma}{4}$ from any point in $\cS_{\{1, 2\}}$. It is easy to see that
all remaining points, $\cS^+$, can only belong to one class, i.e., $\cS^+\cap
X_{\{1,2\}}=\emptyset$ (because all points that did not make it to $\cS_{\{1,
2\}}$ have to be closer than $\tfrac{\gamma}{4}$ from some point in $\cS_{\{1,
2\}}$). Since $\cS^+$ does not have overlapping classes, we can execute
\fMMSwap (Algorithm~\ref{algo:binary}) on it to select a set with $k_i - t$
points for each class i in $\{1, 2\}$. In our analysis, we show that $\cS^{+}$
contains at least $k_1-t$ and $k_2-t$ points from $X_{\{1\}}$ and $X_{\{2\}}$
that are $\geq \gamma$ apart from each other. Thus, the \fMMGSwap algorithm
will produce a set of points that are at least $\gamma/4$ apart from each other.

\begin{theorem}\label{thm:swapover}
\fMMGSwap (Algorithm~\ref{algo:swap2}) is a polynomial-time algorithm with $1/4$-approximation guarantee for the fair diversification problem with $m=2$ overlapping classes. 
\end{theorem}	 

\begin{proof}
Let $O=O_{\{1\}}\cup O_{\{2\}} \cup O_{\{1, 2\}}$ be the set that maximizes
diversity and satisfies the fairness constraints. Let $\lsf = \diver{(O)}$, which implies that $d(o_1, o_2) \geq \lsf$ for any pair of optimal elements $o_1, o_2 \in O$. We will show that for any guess $\gamma \leq \lsf$, Algorithm~\ref{algo:swap2} returns a set $\cS = \cS_{\{1\}} \cup \cS_{\{2\}} \cup \cS_{\{1,2\}}$ with $\diver(\cS) \geq \gamma/4$. First, note that by the definition of $\cS_{\{1,2\}}$ set, it holds that:
 
\begin{equation*} \label{eqgen:eq1}
\diver(\cS_{\{1,2\}}) \geq \gamma/4 
\end{equation*}

Next, notice that the $\setB$ set in line~\ref{ln:close} of Algorithm~\ref{algo:swap2} consists of all the points in $X_{\{1, 2\}}$, and any single-colored points $<\gamma/4$ apart from some point in $\cS_{\{1,2\}}$. As a result, we know that: (1)~all the points remaining in $\setA$ are greater or equal than $\gamma/4$ apart from all the points in $\cS_{\{1,2\}}$, and (2)~$\setA$ only contains single-colored points (if there were any bi-colored elements $\geq \gamma/4$ apart from the points in $\cS_{\{1,2\}}$, they would have been added to $\cS_{\{1,2\}}$). We express $\setA$ as follows:  

\[\setA= \setA_{\{1\}} \cup \setA_{\{2\}} \ \text{with} \ \setA_{\{i\}} \subseteq X_{\{i\}} \ \text{for} \ i \in \{1, 2\}\]

Let $t= |\cS_{\{1,2\}}|$. We argue that for any guess $\gamma \leq \lsf$, $\setA$ contains at least $k_i -t$ elements for $i \in \{1, 2\}$ that are $\geq \gamma$ apart. Thus, \fMMSwap will be able to find a set of points to satisfy the fairness constraints that are at least $\gamma/ 4$ apart. Define $c^{-}_{\{1\}}, c^{-}_{\{2\}}, c^{-}_{\{1, 2\}}$ to be the number of  optimal points in $O_{\{1\}}, O_{\{2\}}$ and $O_{\{1, 2\}}$ present in 
$\setB$. Now, notice that: 

\[c^{-}_{\{1\}} + c^{-}_{\{2\}} + c^{-}_{\{1, 2\}} \leq t\]

This holds because at most one optimal point point can be $< \gamma / 4$ from a point in $\cS_{\{1, 2\}}$. Suppose that there exist a pair of optimal points $o_1, o_2 \in O$, and a point $x \in \cS_{\{1, 2\}}$ such that  
$d(o_1, x) < \gamma / 4$ and $d(o_2, x) < \gamma / 4$. Then we derive a contradiction by applying the triangle inequality as: $d(o_1, o_2) \leq d(o_1, x) +d(x, o_2) < \gamma/2 < \lsf/2$. Consequently, it now follows that $\setA$ contains at least $k_1 - c^{-}_{\{1\}}- c^{-}_{\{1, 2\}} \geq k_1 -t$ optimal points of $O_{\{1\}}$, and $k_2 - c^{-}_{\{2\}} - c^{-}_{\{1, 2\}} \geq k_2 -t$ of $O_{\{2\}}$, which by definition of $O$ are greater or equal than $\gamma$ apart. 

So \fMMSwap will be able to find a set $\cS_{\{1\}} \subseteq \setA_{\{1\}}$ and $\cS_{\{2\}} \subseteq \setA_{\{2\}}$ with the required number of elements such that $\diver(\cS_{\{1\}} \cup \cS_{\{2\}}) \geq \gamma/4$. Thus, we get that $\diver(S) \geq \gamma/4$. If we perform a binary search over all the pairwise distances of the points in $\cU$, we will find a guess $\gamma = \lsf$, which implies the claimed approximation factor for \fMMGSwap. 
\end{proof}

\begin{figure}[t]
	\centering
\scalebox{0.9}{
	\begin{minipage}[t]{\textwidth}
		\Genswap
	\end{minipage}}
\end{figure}	

\subsection{Fair-and-Diverse Selection (Overlaps): {\large\boldmath{$m\geq 3$}}}
\label{sec:flowover}

The algorithm in this section is an extension of \fMMFlow (Algorithm ~\ref{algo:m3}); the previous algorithm did not apply in the case when classes could overlap whereas the new algorithm will. Throughout this section, it will be convenient to use the following notation: $M:=\binom{m}{\lfloor{m/2}\rfloor}$. The approximation factor for the algorithm designed in this section will be $3M-1$ in contrast to the $3m-1$ approximation for the non-overlapping case. Note that for $m=2, 3, 4,  5$ we have $M=2, 3, 6, 10$, i.e., when the number of classes is small, $M$ is still relatively small. 

There are two main steps that need to be changed in the overlapping case: 1) defining a subset $Z$ of the elements that will be considered and 2) determining how many points to use that appear in multiple classes. We discuss each in turn.

\medskip
\noindent
\textbf{Defining $Z$.} Recall that the first main part of \fMMFlow (Algorithm~\ref{algo:m3}) was to select a subset of points of each color such that all points in each subset was a certain distance apart. When there are overlapping classes, we need to revisit how this is done. Motivated by the fact that an element in $X_{L'}$ contributes to at least as many fairness constraints as an element in $X_L$ if $L\subset L'$, when we select a subset of points in $\cU_i$ we want to prioritize points that are also
in other classes. 

For example, for $m=3$ we have: 
\begin{eqnarray*}
\cU_1 &=& X_{\{1\}}\cup X_{\{1,2\}}\cup X_{\{1,3\}}\cup X_{\{1,2,3\}}\\
\cU_2 &=& X_{\{2\}}\cup X_{\{1,2\}}\cup X_{\{2,3\}}\cup X_{\{1,2,3\}}\\
\cU_3 &=& X_{\{3\}}\cup X_{\{1,3\}}\cup X_{\{2,3\}}\cup X_{\{1,2,3\}} \ . 
\end{eqnarray*}

Consistent with ``prioritizing points'' in multiple classes, we construct
subsets of $\cU_1, \cU_2, \cU_3$ by first constructing a maximal subset
$Z_{\{1,2,3\}}\subset X_{\{1,2,3\}}$ such that the pairwise distance of all
points is at least $d_1$. We then define a maximal subset $Z_{\{1,3\}} \subset
X_{\{1,3\}}$ such that every point is at least $d_1$ from each other point in
$Z_{\{1,3\}}$ and from points in $Z_{\{1,2,3\}}$. We construct $Z_{\{1,2\}}$
and $Z_{\{2,3\}}$ similarly. Finally $Z_{\{1\}}$ is a maximal subset of
$X_{\{1\}}$ such that every point is at least $d_1$ from each other point in
$Z_{\{1\}}$ and from every point in $Z_{\{1,2\}}\cup Z_{\{1,3\}}\cup
Z_{\{1,2,3\}}$. Lines~\ref{ln:Z}--\ref{ln:loop}
in Algorithm~\ref{algo:general} generalize this process
to arbitrary $m$.

\looseness-1
Note that we ensure the property that all points in $Z_L$ are at least $d_1$
far from each other and from any point in $\bigcup_{L':L\subset L'} Z_{L'}$ but
the subset of elements picked from $\cU_1$, i.e., $Z_{\{1\}}\cup
Z_{\{1,2\}}\cup Z_{\{1,3\}}\cup Z_{\{1,2,3\}} \subset \cU_1$, no longer
satisfies the condition that they are all at least $d_1$ far from one another.
In particular, there may exist points $x\in Z_L$ and $y\in Z_{L'}$ such that
$d(x,y)<d_1$ if neither $L$ or $L'$ is a subset of the other.\footnote{This is
a generalization of the case when there was no-overlap. In that case there
could exist $x\in Z_i$ and $y\in Z_j$ such that $d(x,y)<d_1$.} A natural
question, and an issue that will arise in our analysis is how many sets can
there be such that no set is a subset of another. Fortunately, the following
classic result in extremal combinatorics resolves this question.

\begin{figure}[t]
	\centering
\scalebox{0.9}{
	\begin{minipage}[t]{\textwidth}
		\GenFlow
	\end{minipage}}
\end{figure}	

\begin{lemma}[Sperner's Lemma] \label{def:anti-chain}\label{lem:sperner}
A collection of sets is called an \emph{anti-chain} if none of the sets is a subset of another set. If all sets are subsets of $[m]$ then the maximum size of such a  collection is $M=\binom{m}{\lfloor m/2 \rfloor}$.
\end{lemma}

Next, recall that \fMMFlow (Algorithm \ref{algo:m3}) then constructs a graph
$G_Z$ where the nodes are the selected points and there are edges between
points if this distance is $<d_2$. The new algorithm proceeds similarly but
with new parameters: $d_1\leftarrow \frac{\comb \gamma}{3\comb-1}$, and
$d_2\leftarrow \frac{\gamma}{3\comb-1}$. With this setting of the parameters
and appealing to Lemma \ref{lem:sperner} we prove an upper bound on the
distance between any two points in the same connected components: 

\begin{lemma}\label{lem:groupwidthGen}
For all connected components $C_j$, \[\forall x,y \in C_j:~d(x,y) < \left(\comb-1\right)d_2,\]
and $C_j$ does not contain any two points $a,b$ such that $a\in X_L$ and $b\in X_{L'}$ where $L\subset L'$. 
\end{lemma}

\begin{proof}
Consider two points $x,y\in C_j$ and let the length of a shortest unweighted path $P_{x,y}$ between $x$ and $y$ in the graph be $\ell$. If $\ell\leq \comb-1$ then $d(x,y)< (\comb -1)d_2$ as required. If $\ell\geq \comb$ then
by Lemma~\ref{lem:sperner}, there must exist two points on this path (including end points) in $X_L$ and $X_{L'}$ such that $L$ and $L'$ are \emph{comparable}, i.e., $L$ is a subset of $L'$ or vice versa and  this will lead to a contradiction. Consider the  subpath $P_{x',y'}\subset P_{x,y}$ such that $x'\in X_L$ and $y'\in X_{L'}$ for some comparable $L$ and $L'$. If the internal nodes are $x_1, x_2, \ldots $ and these belong to sets $X_{L_1}, X_{L_2} ,\ldots$ then by definition of $x'$ and $y'$, the collection of sets $\{L_1, L_2, \ldots, L'\}$ is  an anti-chain and hence the size of this collection is at most $\comb$ by Lemma \ref{lem:sperner}. Hence, the length of the path between $x'$ and $y'$ is also at most $\comb$ and therefore
\[d(x',y')< \comb d_2=d_1\] 
But this contradicts $d(x',y')\geq d_1$ because $x'\in X_{L}$ and $y'\in X_{L'}$ where $L$ and $L'$ are comparable.
\end{proof}

\medskip
\noindent
\textbf{Guessing how much to exploit points in multiple classes.} So far we have (1)~discussed how to select the subset $Z$ of input points and (2)~partitioned $Z$ such that we have some upper bound on the distance between
any two points in the same partition. In the non-overlapping case, we could
then argue it suffices to pick at most one point in each partition and adding
this point to the output set $\cS$ would increment $|\cS \cap \cU_i|$ for
exactly one value $i\in [m]$. In the overlapping case, however, we may need to
pick a point in a partition that is in multiple classes and would increment
$|\cS \cap \cU_i|$ for multiple values of $i$.

To get the reduction to network flow to generalize to the non-overlapping case
we need to guess values $c_{L}$ for every non-empty set $L \subset [m]$ and
require that we find at least $c_{L}$ points in $\cap_{i\in L} \cU_i$ such that
the $\sum_{L \subseteq [m]} c_{L}$ points returned are distinct. The fact the
points need to be distinct allows the reduction to go through. Note that to
satisfy the fairness requirements we need that $\sum_{L:i\in L} c_L\geq k_i$
for each $i$. 

\begin{example} 
    Suppose we require $k_1=2$ points from $\cU_1$ and $k_2=2$ points from
    $\cU_2$. Then the guess $c_{\{1\}}=2$ and $c_{\{2\}}=2$ would correspond to
    picking at least four distinct points, at least two from $\cU_1$ and at
    least two from $\cU_2$. In contrast, the guess $c_{\{1\}}=c_{\{2\}}=1$, and
    $c_{\{1,2\}}=1$ would correspond to picking at least three distinct points
    where at least one comes from each of sets $\cU_1, \cU_2, \cU_1\cap \cU_2$
    respectively.
\end{example}

There are at most  $k^{2^m-1-m}$ possible guesses\footnote{Recall that we typically consider $m$ to be a small constant. A bound of $k^{2^m-1}$ is immediate because there at most $2^m-1$ quantities. A slightly tighter bound follows by noting that $c_L$ for all singleton sets $L$ is implied once the other values are chosen.}  to try for the values and at least one is feasible since the optimal solution corresponds to some set of guesses. With a feasible set of guesses, we then essentially treat all sets $L\subseteq [m]$ as colors although when we need to pick  $c_L$ points of color $L$, it will suffice to pick points with color $L'$ if $L'$ is a subset of $L$.

The next theorem establishes that when the algorithm does not abort, the solution returned has diversity at least $\gamma/(3\comb-1)$ and that it never aborts if the guess $\gamma$ is at most the optimum diversity.

\begin{theorem}\label{thm:genFlow}
Let $\lsf$ be the optimum diversity. If $\gamma\leq \ell^*_{\text{fair}}$ then the algorithm returns a set of points of the required colors that are each $\geq \gamma/(3\comb-1)$ apart.  If $\gamma> \ell^*_{\text{fair}}$ then the algorithm either aborts or returns a set of points of the required colors that are each $\geq d_2=\gamma/(3\comb-1)$ apart.  
\end{theorem}

\begin{proof}
Note that if the algorithm does not abort  then all points are $\geq \gamma/(3\comb-1)$ apart since any two points in different connected components are $\geq \gamma/(3\comb-1)$ apart. 

Hence, it remains to argue that if $\gamma\leq \lsf$ then the algorithm does not abort. To argue this, we will show it is possible to construct a flow of size $\sum c_L$. And to do this it suffices to, for each $L \subset [m]$, identify $c_L$ different connected components that each include a point from $\cap_{i\in L} \cU_i$. 

Let $O = \bigcup_{L\subset [m]} O_L$ be an optimal solution where $O_L=O\cap X_L$ and let $c_L=|O_L|$. We will henceforth consider the iteration of the algorithm which guessed this set of $\{c_{L}\}_{L \subset [m]}$ values. For every point $x \in O$, let $f(x)$ be the closest point in $Z$ where for all $i$, \[x\in \cU_i \Rightarrow f(x)\in \cU_i\]

Note that this requirement ensures that if $x$ is replaced by $f(x)$ then all the fairness constraints are still satisfied. By construction of $Z$, $d(x,f(x))< d_1$. Hence, for any $x,y\in O$, 
\[d(x,y)> \lsf-2d_1 \geq \gamma-2\gamma \comb/(3\comb-1) =(\comb-1)d_2\]
and hence, by Lemma \ref{lem:groupwidthGen}, this implies that all points in 
$f(O)$ are in different connected components. This implies that  there  exist connected components with the necessary requirements. 
\end{proof}

\looseness-1
The rest of the algorithm and analysis follows similarly as Algorithm~\ref{algo:m3}, where we binary search for $\gamma$ in either a continuous or discrete space.  The running time is increased by a factor of $k^{2^m-m-1}$ because of the need to guess the values $\{c_L\}_{L \subset [m]}$; thus \fMMGFlow is a polynomial-time algorithm with a $\tfrac{1}{3 \binom{m}{\lfloor m/2\rfloor}-1}$-approximation guarantee.
	\section{Related Work} \label{sec:relWork} 
Diversity is an important principle in data selection and summarization,
facility location, recommendation systems and web search. The diversity models
that have been proposed in the literature can be organized into three main
categories, (1)~the distance-based models where the goal is to minimize the
\emph{similarity} of the elements within a set, (2)~the coverage-based models
where there exists a predetermined number of categories and the aim is to
maximize the \emph{coverage} of these categories~\cite{10.1145/1498759.1498766,
Munson2009SidelinesAA} and (3)~the novelty-based models that are defined so as
to minimize the \emph{redundancy} of the elements shown to the
user~\cite{10.1145/290941.291025}. For further information, we refer the reader
to the related surveys~\cite{article, Drosou:2010:SRD:1860702.1860709}.

Max-Min and Max-Sum diversification are two of the most well studied
distance-based models~\cite{Chandra:2001:AAD:374591.374596,
10.1145/1526709.1526761, Hassin:1997:AAM:2308449.2308622,
1994:HSC:2753204.2753214}, and there exist efficient algorithms with strong
approximation guarantees for the unconstrained version of the problems in the
offline setting (discussed in Sections~\ref{sec:intro}
and~\ref{sec:Background}). The problem of diversity maximization has also been
studied in the streaming and distributed settings, where (composable) core-sets
were shown to be a useful theoretical tool~\cite{Aghamolaei2015DiversityMV,
10.14778/3055540.3055541, 10.1145/2594538.2594560}, and more recently in the
sliding window setting~\cite{10.1145/3294052.3319701}. A separate line of work
focuses on designing efficient indexing schemes for result
diversification~\cite{10.1145/3375395.3387667,
6477041,10.14778/3192965.3192969}; this direction is orthogonal to our work,
and it is not clear how to extend existing indexing schemes for fair Max-Min
diversification.

There is relatively little prior work on constrained diversification. The
closest to our work is fair Max-Sum diversification (discussed in
Section~\ref{sec:intro}) and fair $k$-center
clustering (discussed in Section~\ref{sec:intro} and Appendix~\ref{sec:appTh}). To the best of our knowledge, our work is the first to augment the traditional Max-Min objective with fairness constraints. 

Prior work has also combined fairness with the determinant measure of diversity~\cite{pmlr-v80-celis18a}. That work models fairness constraints the same way as we do, but their algorithmic framework is entirely different. There, data is represented as vectors, and at each iteration the algorithm identifies the item that is most orthogonal to the current vector, which gets updated with the new item's projection. The limitation of this method is that it can only work in high-dimensional data (e.g., it would not work at all on one-dimensional data). Other work on diverse set selection focused on satisfying fairness constraints while optimizing an additive utility~\cite{Stoyanovich2018OnlineSS}. These methods do not apply to
our setting as Max-Min is not additive. Prior work has also examined the satisfaction of fairness constraints or preferences in specialized settings, such as rankings~\cite{Celis2017RankingWF, Yang:2019:BRD:3367722.3367886, YangS2017}. Work in this domain focuses on specifying and measuring fairness and augmenting ranking algorithms with fairness considerations. Related work on diverse top-$k$ results focuses on returning search results by a combined measure of relevance and dissimilarity to results already produced~\cite{AngelK2011, QinYC2012}.

Our fairness constraints are based on the definitions of \emph{group fairness}
and \emph{statistical parity}~\cite{Dwork:2012:FTA:2090236.2090255}. We do not
pick a particular definition of fairness, and do not place particular
restrictions on the values and distribution of $\langle k_1,\dots,k_m\rangle$.
This model can express equal and proportional representation, as well as any
other distribution. There are other, non-parity-based definitions of fairness
that fall outside our framework. For example, \emph{individual or causal
fairness}~\cite{GalhotraBM2017} examine differences in treatment of individuals
from different groups who are otherwise very similar, but these are not the
focus of this work.
	\section{Summary and Future Directions} \label{sec:conclusions} 
In this paper, we focused on the problem of diverse data selection under
fairness constraints. 
To the best of our knowledge, our work is the first to
introduce fairness constraints to Max-Min diversification. We studied both
cases of disjoint and overlapping groups and proposed novel polynomial
algorithms with strong approximation guarantees. For the case of disjoint
groups, our algorithms have linear running time with respect to the size of the
data. Overall, our work augments in significant ways the existing literature of traditional problems that have been studied under group fairness constraints. We discuss here some possible directions that extend our work through the
exploration of problem variants, or intuitions towards improvement of the known
algorithms and bounds.

\smallskip
\noindent 
\textbf{Improved bounds.}
An interesting open question is whether an $\tfrac{1}{2}$ approximation for
\fMM is possible, as is the case for Max-Min and fair Max-Sum diversification.
In Section~\ref{sec:2.2}, we discussed the correspondence between
fairness constraints and partition matroids. It is possible that results
relevant to matroids can be exploited to improve the algorithms and bounds for
the \fMM problem.

\smallskip
\noindent 
\textbf{Extending the swap algorithm to the general case.} 
Our \fMMSwap algorithm provides a better bound compared to our \fMMFlow
algorithm for the case of $m=2$ ($\tfrac{1}{4}$ and $\tfrac{1}{5}$
respectively). This indicates the possibility that the swap algorithm, if
extended to the general case, could perhaps result in a better bound than
\fMMFlow.

\smallskip
\noindent 
\textbf{Problem variants.}
Our algorithms aim to approximate the diversity score of the optimal solution
to \fMM, while guaranteeing the satisfaction of the fairness constraints. A
possible problem variant could explore the relaxation of the fairness
constraints, and seek to minimize their violation while guaranteeing a
diversity score at least as good as the solution to unconstrained Max-Min
diversification.
Another interesting future direction is to study the fair variant of other
diversity objectives proposed in the literature~\cite{Chandra:2001:AAD:374591.374596, 10.1145/2594538.2594560},
for which there are currently no known results.

\subsection*{Acknowledgements} This work was supported by the NSF under grants  CCF-1934846, CCF-1908849, CCF-1637536, IIS-1453543, CCF-1763423, and IIS-1943971.

	\bibliographystyle{abbrv}
	\bibliography{literature}
	
	\appendix
	\section*{Appendix}
	\section{Results on fair $k$-center clustering} \label{sec:appTh}
In this paper, our primary focus has been on fair diversification using the
Max-Min objective. In Section~\ref{sec:intro}, we discussed how the
unconstrained Max-Min diversification and the $k$-center clustering problems
are closely-related; notably, the best approximation algorithms for the
unconstrained variants of both problems are essentially equivalent and result
in the same approximation bound. In this section, we formally define the $k-$
center clustering, introduce its fair variant and discuss the known
approximation results for this problem. We then explore how algorithms and
intuitions from our work on fair Max-Min diversification can be adapted towards
the fair $k$-center clustering problem.

\smallskip

\noindent
\textbf{The $k$-center and fair $k$-center clustering problems.} 
The objective of $k$-center clustering is to identify $k$ cluster centers, such
that the maximum distance of any point in the universe of elements
$\mathcal{U}$ from its closest cluster center is minimized. This maximum
distance is referred to as the \emph{clustering radius}. More formally, given a
distance metric $d$, $k$-center clustering is expressed by the following
minimization problem: \[\underset{\mathcal{S}\subseteq\mathcal{U}, |\mathcal{S}|=k}{\text{minimize}} \;\; \max_{u\in\mathcal{U}}d(u,\mathcal{S}) \] 
where $d(u,\mathcal{S}) = \min_{s\in\mathcal{S}}d(u,s)$.
Note that this objective does not preclude cluster centers from being close to each other, and in fact an optimal solution to $k$-center clustering could be arbitrarily bad for Max-Min diversification.

\smallskip

\noindent
\textbf{Algorithms and approximations.}
Just like Max-Min diversification, $k$-center clustering is NP-complete. The
greedy approximation algorithm proposed by
Gonzalez~\cite{Gonzalez1985ClusteringTM} is essentially equivalent to \gonzalez
(Algorithm~\ref{algo:GMMA}) and provides a 2-approximation with linear running
time.  

Notably, there is recent work that augments the problem with fairness
constraints~\cite{pmlr-v97-kleindessner19a}: Given $m$ non-overlapping classes
in $\mathcal{U}=\cup_{i=1}^m\mathcal{U}_i$ and non-negative integers $\langle
k_1, \dots, k_m\rangle$, the goal is to derive a set of cluster centers
$\mathcal{S}$, such that $|\mathcal{S}\cap\mathcal{U}_i|=k_i$. The fair
k-center clustering problem can also be expressed by a partition matroid, for
which Chen et al.~\cite{Chen2016} provide a 3-approximation with a quadratic
runtime. Kleindessner et al.~\cite{pmlr-v97-kleindessner19a} provide a
linear-time 5-approximation algorithm for the case of two classes ($m=2$), and
a linear-time $\left(3\cdot 2^{m-1}-1\right)$-approximation for the general
case, a result recently improved to $3(1+\epsilon)$ by Chiplunkar et
al.~\cite{chiplunkar2020solve} and to 3-approximation by Jones et al.~\cite{Jones2020FairKV}.

\smallskip

\noindent
In Section~\ref{sec:clustering}, we adapt the flow algorithm for fair Max-Min diversification, and provide a linear-time $3$-approximation for fair $k$-center clustering. (noting that the three results were derived independently.)

\subsection{Fair $k$-center clustering}\label{sec:clustering}

We show how we can adapt our \fMMFlow algorithm (Algorithm~\ref{algo:m3}) and
design a constant factor 3-approximation for fair $k$-center clustering with
linear running time. 

\smallskip
\noindent
\textbf{Basic algorithm.} 
We start by presenting a basic algorithm that takes as input a guess $\gamma$
for the optimum fair clustering radius. If this guess is less than the optimum
fair clustering radius $\rsf$ then the algorithm may abort but otherwise it
will return a fair clustering with radius at most $3\gamma$.

\smallskip

\noindent
\textbf{Algorithm and intuition.} 
The basic idea behind \fClust (Algorithm~\ref{algo:clustering}) is to construct
a set of points $Y=\{y_1, \ldots, y_t\}$ where all distances between these
points are $> 2\gamma$ apart and all points not in this set are $\leq 2\gamma$
from some point in $Y$; this can be done via the \gonzalez algorithm
(lines~\ref{ln:clustGMM} and~\ref{ln:clustT}). The fact that each pair is $>
2\gamma$ apart implies that any $k$-center clustering, fair or otherwise, with
covering radius $\leq \gamma$ has the property that at least one center must be
within a distance $\gamma$ from each $y_i$ and that no center is within
distance $\gamma$ of two points $y_i, y_j$ since, by appealing to the triangle
inequality, this would violate the fact that $d(y_i,y_j)>2\gamma$.

The algorithm constructs a sets $C_1, \ldots , C_t$ such that we will be able
to argue that if we can pick a fair set of cluster centers from $C_1\cup \ldots
\cup C_t$ such that \emph{exactly} one point is picked in each $C_j$ then we
get a clustering with cluster radius $3\gamma$. Furthermore, if $\gamma\geq
\rsf$, such a set of centers can be proven to exist. We will then be able to
find these centers via a reduction to network flow. The network constructed is
the same as in Algorithm \ref{algo:m3} although the $C_j$ sets in that
algorithm are constructed differently. The only difference is that because we
need exactly one point in each of $C_1, C_2, \ldots, C_t$, we need to find a
flow of size $t$ rather than a flow of size $k$. Note that if we are able to
construct a flow of $t \leq k$, we can arbitrarily add the cluster centers
missing from a class $i \in [m]$ without affecting the clustering radius of the
solution.

\newcommand{\Clust}{
	\begin{algorithm}[H] 
		\caption{\fClust: Fair $k$-Center Clustering}\label{algo:clustering} 
		{\small
			\begin{algorithmic}[1]
				\Statex 
				\begin{description} 
					\item[\rlap{Input:}\phantom{Output:}]
					$\cU_1, \ldots, \cU_m$: Universe of available elements 
					\item[\phantom{Output:}] $k_1, \ldots, k_m \in \mathbb{Z}^{+}$
					\item[\phantom{Output:}] $\gamma\in \mathbb{R}$: A guess of  					optimum fair clustering radius.
					\item[Output:] $k_i$ points in $\cU_i$ for each $i\in [m]$ 
					\end{description} 
					\Procedure{\fClust}{}
					\State $Y=\{y_1, \ldots, y_{k+1}\} 					\leftarrow	\gonzalez(\cU,\emptyset, k+1)$ \label{ln:clustGMM}
					
					\For{$j\in [k]$}
						\State $D_j \leftarrow \{\argmin_{x\in \cU_i} d(x,y_j): i\in 						[m]\} $
					\EndFor
					
					\If{$d(y_{k+1}, \{y_1,\ldots, y_k\})>2\gamma$}
						\Return $\emptyset$ \Comment{Abort} \label{ln:ClustAbort}
					\Else
						\State $Y=\{y_1, \ldots, y_t\}$ with minimum  $t\leq k$ such 						that \label{ln:clustT}
						\Statex\hspace{\algorithmicindent}\hspace{5em} $d(y_{t+1}, 						\{y_1,\ldots, y_t\})\leq 2\gamma$  
					\EndIf
					
					\For{$j\in [t]$}
						\State $C_j \leftarrow \{x\in D_j: d(x,y_j)\leq \gamma\}$
					\EndFor
					\LineCommentx{Construct flow graph}
					\State Construct directed graph $G=(V,E)$ where
					\begin{eqnarray*}
						V&=& \{a,u_1, \ldots, u_m, v_1, \ldots, v_t, b\}  \\
						E&=& \{(a,u_i) \mbox{ with capacity $k_i$}: i \in [m]\} \\
						& & \cup~ \{(v_j,b) \mbox{ with capacity $1$}: j \in [t]\} \\
						& & \cup~  \{(u_i,v_j) \mbox{ with capacity $1$}:  |Z_i\cap 						C_j|\geq 1\} 
						\end{eqnarray*} 

					\State Compute max $a$-$b$ flow. 
					\If{flow size $<t$}
						\Return $\emptyset$ \Comment{Abort}
					\Else \Comment{max flow is $t$}
						\State $\forall(u_i,v_j)$ with flow add a node in $C_j$ with 						color $i$ to $\mathcal{S}$
					\EndIf
						\Return $\mathcal{S}$
					\EndProcedure
			\end{algorithmic}
		}
\end{algorithm}
}

\begin{figure}[t]
	\centering
	\scalebox{0.85}{
	\begin{minipage}[t]{\textwidth}
		\Clust
	\end{minipage}}
\end{figure}	

\begin{theorem}
If $\gamma\geq \rsf$ then the above algorithm returns a fair clustering with radius at most $3\gamma$. If $\gamma<\rsf$ then either the algorithm aborts or it returns a fair clustering with radius at most $3\gamma$. 
\end{theorem}

\begin{proof}
Note that if the algorithm does not abort, the algorithm identifies exactly one point in each of the disjoint sets $C_1, \ldots, C_t$ such that at most $k_i$ points of color $i$ are chosen for each color $i\in [m]$. Since the algorithm did not abort at Step 3 we know that all points in $\cU$ are within distance $2\gamma$ of some point $y_i$ and hence at most distance $2\gamma+\gamma$ from the selected point in $C_i$. Hence, we return a fair clustering with covering radius at most $3\gamma$ as required.

It remains to show that if  $\gamma\geq \rsf$ then the algorithm does not abort. The algorithm does not abort at line~\ref{ln:ClustAbort} since this would imply there exist $k+1$ points that are $>2\gamma$ from each other and this implies $\rsf>\gamma$. Define $E_j=\{x:d(x,y_j)\leq \gamma\}$ and note that the optimum solution must pick a point in each $E_j$ since otherwise $y_j$ is not covered within distance $\gamma$. Hence, we know it is possible to pick at most $k_j$ points of color $j$ such that exactly one point $c_j$ is picked in each $E_j$.  Note that $E_j$ has a point of color $i$ iff $C_j$ has a point of color $i$. Hence, it is also possible to pick at most $k_i$ points of color $i$ (for each $i\in [m]$) such that exactly one point $c_j$ is picked in each $C_j$. Hence, there exists a flow of size $t$ where $(u_i,v_j)$ has flow 1 iff $c_j$ has color $i$ and all edges into $b$ are saturated.
\end{proof}

\smallskip

\noindent
\textbf{Final algorithm.}
We now proceed as in the case of \fMMFlow (Section~\ref{sec:flowAlg}): we can
either binary search for the good $\gamma$ over the continuous range
$[d_{\min}, d_{\max}]$ or over the discrete set of all distances between points
in $Y\cup D_1 \cup D_2 \cup \ldots \cup D_m$. In the first case, we need
$O(\log \log_{1+\epsilon} d_{\max}/d_{\min})$ instantiations of the basic
algorithm before we find a clustering with approximation ratio $3(1+\epsilon)$.
In the second case, we need to sort $O(k^2 m^2)$ distances and then need
$O(\log k)$ instantiations.

\begin{theorem}
There is a 3-approximation for fair $k$-center clustering with running time $O(kn+m^2 k^2 \log k)$.
\end{theorem}

\begin{proof}
Note that $Y$ and $D_1, D_2, \ldots, D_m$ can be computed in $O(kn)$ time. The
flow instance has $O(k)$ nodes and $O(mk)$ edges. Hence, it can be solved in
$O(mk^2)$ time \cite{orlin,Orlin13}. The total running time is therefore
$O(kn+m^2 k^2 \log k+mk^2 \log k)$ as required.
\end{proof}

Prior work has proposed another 3-approximation algorithm for fair $k$-center
clustering~\cite{Chen2016}, but the running time of that algorithm is quadratic
with respect to the data size ($n$). Our 3-approximation algorithm for \fClust adapts the ideas for fair Max-Min diversification, offering an alternative approach to the recent linear-algorithm with a $3(1 +\epsilon)$-approximation guarantee by Chiplunkar et al.~\cite{chiplunkar2020solve}, and to the linear-algorithm with a $3$- approximation guarantee by Jones et al.~\cite{Jones2020FairKV}.

\end{document}